\def\confver{0}

\documentclass[11pt]{article}

\usepackage{amsmath, amssymb, amsthm, amsfonts}
\usepackage{accents}
\usepackage{setspace}
\usepackage{hyperref}

\usepackage{appendix}
\usepackage{graphicx}
\usepackage{multirow}
\usepackage{ragged2e}
\usepackage{algorithm}
\usepackage{algpseudocode}
\usepackage{paralist}

\usepackage{framed}

\usepackage{array}

\usepackage[framemethod=tikz]{mdframed}

\usepackage{amsmath}
\usepackage{tikz}
\usepackage{mathdots}
\usepackage{yhmath}
\usepackage{cancel}
\usepackage{color}
\usepackage{siunitx}
\usepackage{array}
\usepackage{multirow}
\usepackage{amssymb}
\usepackage{tabularx}
\usepackage{extarrows}
\usepackage{booktabs}
\usetikzlibrary{fadings}
\usetikzlibrary{patterns}
\usetikzlibrary{shadows.blur}
\usetikzlibrary{shapes}

\ifnum\confver=0
\usepackage[letterpaper,hmargin=1.0in,vmargin=1.0in]{geometry}
\newtheorem{theorem}{Theorem}[section]
\newtheorem{corollary}{Corollary}[theorem]

\newtheorem{definition}[theorem]{Definition}

\newtheorem{claim}[theorem]{Claim}

\fi 

\newcommand{\Vor}{{\rm Vor}}
\newcommand{\Cen}{{\rm Cen}}
\newcommand{\eps}{\epsilon}
\newcommand{\calC}{{\cal C}}

\newcommand{\poly}{{\rm poly}}
\newcommand{\polylog}{{\rm polylog}}
\newcommand{\dist}{{\rm d}}

\newif\ifdraft
\draftfalse

\newcommand{\tO}{\tilde{O}}

\newcommand{\eqdef}{\stackrel{\rm def}{=}}

\newcommand{\allnbr}{\mathtt{ALL\_NBR}}

\newif\ifnames
\namestrue

\title{Improved LCAs for constructing spanners}
\ifnames{
\author{Rubi Arviv\thanks{Efi Arazi School of Computer Science, Reichman University. Email: {\tt rubi.arv@gmail.com}.}
\and Lily Chung \thanks{MIT. Email: {\tt lkdc@mit.edu}. Supported by an Akamai Presidential Fellowship.}
\and 
Reut Levi\thanks{Efi Arazi School of Computer Science, Reichman University. Email: {\tt reut.levi1@runi.ac.il}.}
\and Edward Pyne \thanks{MIT. Email: {\tt epyne@mit.edu}. Supported by an Akamai Presidential Fellowship.}
}
}\fi

\title{Improved Local Computation Algorithms for Constructing Spanners}

\begin{document}
\begin{titlepage}
\maketitle
\begin{abstract}
    A spanner of a graph is a subgraph that preserves lengths of shortest paths up to a multiplicative distortion. For every $k$, a spanner with size $O(n^{1+1/k})$ and stretch $(2k+1)$ can be constructed by a simple centralized greedy algorithm, and this is tight assuming Erd\H{o}s girth conjecture.
 
    In this paper we study the problem of
    constructing spanners in a local manner, specifically in the Local Computation Model proposed by Rubinfeld et al. (ICS 2011).
    
    We provide a randomized Local Computation Agorithm (LCA) for constructing $(2r-1)$-spanners with $\tO(n^{1+1/r})$ edges
	and probe complexity of $\tO(n^{1-1/r})$ for $r \in \{2,3\}$, where $n$ denotes the number of vertices in the input graph. Up to polylogarithmic factors, in both cases, the stretch factor is optimal (for the respective number of edges).
	In addition, our probe complexity for $r=2$, i.e., for constructing a $3$-spanner, is optimal up to polylogarithmic factors.
	Our result improves over the probe complexity of Parter et al. (ITCS 2019) that is $\tO(n^{1-1/2r})$ for $r \in \{2,3\}$. Both our algorithms and the algorithms of Parter et al. use a combination of neighbor-probes and pair-probes in the above-mentioned LCAs.
	
	For general $k\geq 1$, we provide an LCA for constructing $O(k^2)$-spanners with $\tO(n^{1+1/k})$ edges using $O(n^{2/3}\Delta^2)$ neighbor-probes, improving over the $\tO(n^{2/3}\Delta^4)$ algorithm of Parter et al.
 
    By developing a new randomized LCA for graph decomposition, we further improve the probe complexity of the latter task to be $O(n^{2/3-(1.5-\alpha)/k}\Delta^2)$, for any constant $\alpha>0$. This latter LCA may be of independent interest.
\end{abstract}
\vfill
\thispagestyle{empty}
\end{titlepage}

\section{Introduction}
A spanner is a sparse structure that is a subgraph of the input graph and preserves, up to a predetermined multiplicative factor, the pairwise distance of vertices.
Formally, a $k$-spanner of a graph $G = (V,E)$ is a graph $G' = (V, E')$ such that $E' \subseteq E$, in which the distance between any pair of vertices in $G'$ is at most $k$ times longer than the corresponding distance in $G$. $k$ is referred to as the {\em stretch} of the spanner. 

Spanners have numerous applications in a wide variety of fields such as communication networks~\cite{Aw85, PU87, PU89}, biology~\cite{BD89} and robotics~\cite{Ch86, DFS90}.
Consequently, the problem of constructing spanners has been studied extensively in several models, such as the distributed model~\cite{BKS12, DG08, DGP07, DGPV08, DGPV09, EN17, P10}, streaming algorithms~\cite{AGM12, KW14} and dynamic algorithms~\cite{BK16, BFH19}. 

This problem was also considered in the realm of sublinear algorithms and in particular 
in the model of 
{\em Local computation algorithms} (LCAs) introduced by Rubinfeld et al.~\cite{RTVX11} (see also Alon et al.~\cite{ARVX12} and survey in~\cite{LM17}).
In this model the goal is to avoid computing the entire output and instead to compute parts of the output on demand. This model is suitable for the case that not only the input is massive but also the output. 
Moreover, LCAs support queries from different users while preserving consistency with a single valid solution (although there might be several valid solutions) across different queries. 
The notion of computing the output locally goes back to local algorithms, locally decodable codes and local reconstruction algorithms. LCAs can be viewed as a generalization of these frameworks.

Recently, several works~\cite{LRR20,LRR16,LL18,PRVY19} considered the problem of constructing spanners in the LCA model.
The formulation of the problem in this model is as defined next. 
\begin{definition}[\cite{ARVX12,LRR20}]
An \emph{LCA} $\mathcal{A}$ for graph {\em spanners} is a (randomized) algorithm
with the following properties. $\mathcal{A}$ has access to the adjacency list oracle $\mathcal{O}^G$ of the
input graph $G$, a tape of random bits, and local read-write computation memory. 
When given an input (query) edge $(u,v) \in E$, $\mathcal{A}$ accesses $\mathcal{O}^G$ by making probes, then returns {\rm YES} if $(u,v)$ is in the spanner $H$, or returns {\rm NO} otherwise.
This answer must only depend on the query $(u,v)$, the graph $G$, and the random bits. For a fixed tape of random bits, the answers given by $\mathcal{A}$ to all possible edge queries, must be consistent with one particular sparse spanner.
\end{definition}
For specific details regarding the types of probes supported in the LCA model, we refer the reader to Section~\ref{sec:preliminaries}.

\subsection{Our Results}\label{sec:our}
We provide LCAs that with high probability construct the following spanners.
\begin{enumerate}
    \item \label{item1}A $3$-spanner with $\tO(n^{1+1/2})$ edges. The probe and time complexity of the algorithm is $\tO(n^{1/2})$ which is optimal up to polylogarithmic factors (and constitutes the first optimal algorithm for general graphs). The size-stretch trade-off is optimal as well (up to polylogarithmic factors). This improves over the algorithm of Parter et al.~\cite{PRVY19} whose probe and time complexity is $\tO(n^{3/4})$.
    \item \label{item2} A $5$-spanner with $\tO(n^{1+1/3})$ edges (the size-stretch trade-off is optimal up to polylogarithmic factors). The probe and time complexity of the algorithm is $\tO(n^{2/3})$.
    This improves over the algorithm of Parter et al.~\cite{PRVY19} whose probe and time complexity is $\tO(n^{5/6})$.
    \item \label{item3} An $O(k^2)$-spanner with $\tO(n^{1+1/k})$ edges with high probability. The probe and time complexity of the algorithm is $O(n^{2/3}\Delta^2)$ where $\Delta$ denotes the maximum degree of the input graph.
    This improves over the algorithm of Parter et al.~\cite{PRVY19} whose probe and time complexity is $\tO(n^{2/3} \Delta^4)$. Our algorithm (and the algorithm of~\cite{PRVY19}) uses only neighbor probes for this task.
    \item \label{item4} By additionally taking advantage of adjacency probes we further improve the probe and time complexity of the latter algorithm to be $O(n^{2/3-(1.5-\alpha)/k}\Delta^2)$, for any constant $\alpha>0$. 
    This result utilizes a new, efficient local computation algorithm for decomposing a graph into subgraphs with improved maximum degree that may be of independent interest. 
\end{enumerate}

\subsection{Our algorithms and techniques}
We next describe our algorithms in high-level.
Our LCAs for constructing $3$-spanners and $5$-spanners share similarities with the LCAs in~\cite{PRVY19} (which are inspired by the algorithm of Baswana and Sen for constructing spanners~\cite{BS07}). The main novelty of our algorithms is in selecting several sets of centers, each designed to cluster different type of vertices. The basic idea is that for high-degree vertices we need to select less centers. 
Consequently, we can allow more edges per pair of vertex-cluster or cluster-cluster which decreases the probe complexity. 
To support this approach we also change the way each vertex finds its center. 
See more details in Subsections~\ref{sec:3spanH} and~\ref{sec:5spanH}.

Our algorithm for constructing $O(k^2)$-spanners consists of two parts. The first, which is described in high-level is Subsection~\ref{sec:kspanH}, closely follows the construction in~\cite{PRVY19}. The main novelty in this algorithm is in the way we partition the Voronoi cells, which are formed with respect to randomly selected centers, into clusters of smaller size. In addition, we make other adjustments in order to save an additional factor of $\Delta$ in the probe and time complexity. The second is a new LCA for decomposing a graph into subgraphs of smaller maximum degree. As the first algorithm depends quadratically on the degree, this allows for further savings. We elaborate on this algorithm, which may be of independent interest, in Subsection~\ref{sec:decomp_intro}.

\subsection{Algorithm for constructing \texorpdfstring{$3$}{3}-spanners}\label{sec:3spanH}
We begin with describing our algorithm for constructing $3$-spanners from a global point of view. The local implementation of this global algorithm is relatively straight-forward.

The high level idea is as follows.
We consider a partition of the vertices into {\em heavy} and {\em light} according to their degrees.
All the edges incident to light vertices are added to the spanner.
We now focus on the heavy vertices.
As a first step, a random subset of vertices is selected. We refer to these vertices as {\em centers}. 
 With high probability, every heavy vertex has a center in its neighborhood. 
 Assuming this event occurs, each heavy vertex joins a {\em cluster} of at least one of the centers in its neighborhood. A cluster is composed from a center and a subset of its neighbors. 
 On query $\{u, v\}$, where both $u$ and $v$ are heavy, we consider two cases.
 \begin{compactenum}
     \item $u$ and $v$ belong to the same cluster. In this case we add the edge $\{u, v\}$ to the spanner only in case $u$ is the center of $v$ or vice versa.
    \item Otherwise, $u$ and $v$ belong to different clusters. Assume without loss of generality that the degree of $v$ is not greater than the degree of $u$. 
    We divide the edges incident to $v$ into fixed size {\em buckets} and add the edge $\{u, v\}$ only if it has minimum rank amongst the edges that are incident to the cluster of $u$. 
    See Figure~\ref{fig:3spanner} for an illustration.
 \end{compactenum}
 
   In order to make the above high-level description concrete we need to set up some parameters and describe how the centers are selected and how each vertex finds its center.
   We begin by defining vertices with degrees larger than $\sqrt{n}$ as heavy. Thus by adding all the edges incident to light vertices we add at most $O(n^{3/2})$ edges. 
   
   The selection of the centers proceeds as follows. We define $t = \Theta(\log \sqrt{n})$ sets of centers, which are picked uniformly at random, $S_1, \ldots, S_t$ such that the size of $S_1$ is $\Theta(\sqrt{n})$ and the size of $S_{i+1}$ is roughly half of the size of $S_{i}$. 
  Thus, overall, the number of centers is $\tO(\sqrt{n})$.
  
  We next describe how each heavy vertex finds its center. We partition the heavy vertices into $t$ sets, $V_1, \ldots, V_t$ according to their degrees. The set $V_1$ contains all the vertices with degree in $[\sqrt{n}+1, 2\sqrt{n}]$ and in general for every $i\in [t]$, the set $V_i$ contains all the vertices with degree in $[2^{i-1}\sqrt{n}+1, 2^i\sqrt{n}]$.
  The centers for vertices in the set $V_i$ are taken from the set $S_i$. With high probability, for every $i\in [t]$, each vertex $v \in V_i$ has at least one vertex from $S_i$ in its neighborhood and at most $O(\log n)$. Thus, with high probability, each heavy vertex belongs to at least one cluster and at most $O(\log n)$ clusters.
  
  Given a heavy vertex $v \in V_i$, the centers of $v$ are found by going over all the vertices in $S_i$, $u$, and checking if $\{u, v\}$ is an edge in the graph. Since the total number of centers is $\tO(\sqrt{n})$, the probe and time complexity of finding the center of a given vertex is $\tO(\sqrt{n})$.
  
  It remains to set the size of the buckets. Let $\{u, v\} \in E$ be such that $v \in V_i$ and ${\rm deg}(u) \leq {\rm deg}(v)$. Since $v \in V_i$, it follows that ${\rm deg}(v) \leq 2^i\sqrt{n}$. Since $|S_i| \leq c \sqrt{n}\log n/ 2^i$ for some constant $c$, by setting the size of the buckets to be $\sqrt{n}$ we obtain that the total number of edges between heavy vertices that belong to different clusters is $\tO(n^{3/2})$, as desired.
  
  From the fact that the size of the buckets is $\sqrt{n}$ it follows that the total probe and time complexity of our algorithms is $\tO(\sqrt{n})$. 
  From the fact that the diameter of every cluster is $2$ we obtain that for every edge $\{u, v\}$ which we remove from the graph, there exists a path of length at most $3$ between $u$ and $v$. Hence, the stretch factor of our spanner is $3$, as desired.
  See Figure~\ref{fig:3spanner} for illustration of this part.
  
\ifnum\confver=0
\begin{figure}[H]
\centerline{\scalebox{0.85}{

\tikzset{every picture/.style={line width=0.75pt}} 

\begin{tikzpicture}[x=0.75pt,y=0.75pt,yscale=-1,xscale=1]

\draw  [fill={rgb, 255:red, 74; green, 144; blue, 226 }  ,fill opacity=1 ] (209.82,139) -- (242.05,139) -- (242.05,227) -- (209.82,227) -- cycle ;
\draw [color={rgb, 255:red, 208; green, 2; blue, 27 }  ,draw opacity=1 ][line width=6]    (20.38,165.53) -- (238.11,159.24) ;
\draw [color={rgb, 255:red, 208; green, 2; blue, 27 }  ,draw opacity=1 ][line width=6]    (222.32,218.3) -- (322.96,218.3) ;
\draw [color={rgb, 255:red, 208; green, 2; blue, 27 }  ,draw opacity=1 ][fill={rgb, 255:red, 0; green, 0; blue, 0 }  ,fill opacity=1 ][line width=6]    (238.11,159.24) -- (318.36,218.3) ;
\draw [fill={rgb, 255:red, 0; green, 0; blue, 0 }  ,fill opacity=1 ][line width=0.75]    (238.11,159.24) -- (318.36,218.3) ;
\draw    (238.11,277.36) -- (318.36,218.3) ;
\draw    (318.36,218.3) -- (361,308) ;
\draw  [fill={rgb, 255:red, 0; green, 0; blue, 0 }  ,fill opacity=1 ] (15.78,165.53) .. controls (15.78,162.83) and (17.84,160.64) .. (20.38,160.64) .. controls (22.92,160.64) and (24.98,162.83) .. (24.98,165.53) .. controls (24.98,168.23) and (22.92,170.42) .. (20.38,170.42) .. controls (17.84,170.42) and (15.78,168.23) .. (15.78,165.53) -- cycle ;
\draw [line width=0.75]    (217.72,218.3) -- (318.36,218.3) ;
\draw   (217.72,218.3) .. controls (217.72,162.91) and (262.77,118) .. (318.36,118) .. controls (373.94,118) and (419,162.91) .. (419,218.3) .. controls (419,273.7) and (373.94,318.6) .. (318.36,318.6) .. controls (262.77,318.6) and (217.72,273.7) .. (217.72,218.3) -- cycle ;
\draw  [fill={rgb, 255:red, 0; green, 0; blue, 0 }  ,fill opacity=1 ] (233.5,277.36) .. controls (233.5,274.66) and (235.56,272.47) .. (238.11,272.47) .. controls (240.65,272.47) and (242.71,274.66) .. (242.71,277.36) .. controls (242.71,280.07) and (240.65,282.26) .. (238.11,282.26) .. controls (235.56,282.26) and (233.5,280.07) .. (233.5,277.36) -- cycle ;
\draw  [fill={rgb, 255:red, 0; green, 0; blue, 0 }  ,fill opacity=1 ] (356.4,308) .. controls (356.4,305.3) and (358.46,303.11) .. (361,303.11) .. controls (363.54,303.11) and (365.6,305.3) .. (365.6,308) .. controls (365.6,310.7) and (363.54,312.89) .. (361,312.89) .. controls (358.46,312.89) and (356.4,310.7) .. (356.4,308) -- cycle ;
\draw [color={rgb, 255:red, 0; green, 0; blue, 0 }  ,draw opacity=1 ][line width=0.75]    (20.38,165.53) -- (217.72,218.3) ;
\draw  [fill={rgb, 255:red, 74; green, 144; blue, 226 }  ,fill opacity=1 ] (210.82,290.99) -- (243.05,290.99) -- (243.05,373) -- (210.82,373) -- cycle ;
\draw  [fill={rgb, 255:red, 74; green, 144; blue, 226 }  ,fill opacity=1 ] (209.82,35.99) -- (242.05,35.99) -- (242.05,118) -- (209.82,118) -- cycle ;
\draw    (339,119) -- (318.36,218.3) ;
\draw  [fill={rgb, 255:red, 0; green, 0; blue, 0 }  ,fill opacity=1 ] (313.75,218.3) .. controls (313.75,215.6) and (315.81,213.41) .. (318.36,213.41) .. controls (320.9,213.41) and (322.96,215.6) .. (322.96,218.3) .. controls (322.96,221) and (320.9,223.19) .. (318.36,223.19) .. controls (315.81,223.19) and (313.75,221) .. (313.75,218.3) -- cycle ;
\draw  [fill={rgb, 255:red, 0; green, 0; blue, 0 }  ,fill opacity=1 ] (233.5,159.24) .. controls (233.5,156.54) and (235.56,154.35) .. (238.11,154.35) .. controls (240.65,154.35) and (242.71,156.54) .. (242.71,159.24) .. controls (242.71,161.94) and (240.65,164.13) .. (238.11,164.13) .. controls (235.56,164.13) and (233.5,161.94) .. (233.5,159.24) -- cycle ;
\draw  [fill={rgb, 255:red, 0; green, 0; blue, 0 }  ,fill opacity=1 ] (213.11,218.3) .. controls (213.11,215.6) and (215.17,213.41) .. (217.72,213.41) .. controls (220.26,213.41) and (222.32,215.6) .. (222.32,218.3) .. controls (222.32,221) and (220.26,223.19) .. (217.72,223.19) .. controls (215.17,223.19) and (213.11,221) .. (213.11,218.3) -- cycle ;
\draw    (318.36,218.3) -- (419,218.3) ;
\draw [color={rgb, 255:red, 0; green, 0; blue, 0 }  ,draw opacity=1 ][line width=0.75]    (20.38,165.53) -- (238.11,159.24) ;
\draw  [fill={rgb, 255:red, 0; green, 0; blue, 0 }  ,fill opacity=1 ] (334.4,119) .. controls (334.4,116.3) and (336.46,114.11) .. (339,114.11) .. controls (341.54,114.11) and (343.6,116.3) .. (343.6,119) .. controls (343.6,121.7) and (341.54,123.89) .. (339,123.89) .. controls (336.46,123.89) and (334.4,121.7) .. (334.4,119) -- cycle ;
\draw [color={rgb, 255:red, 0; green, 0; blue, 0 }  ,draw opacity=1 ][line width=0.75]    (20.38,165.53) -- (209,55) ;
\draw [color={rgb, 255:red, 0; green, 0; blue, 0 }  ,draw opacity=1 ][line width=0.75]    (20.38,165.53) -- (210,81) ;
\draw [color={rgb, 255:red, 0; green, 0; blue, 0 }  ,draw opacity=1 ][line width=0.75]    (20.38,165.53) -- (210,300) ;
\draw [color={rgb, 255:red, 0; green, 0; blue, 0 }  ,draw opacity=1 ][line width=0.75]    (20.38,165.53) -- (211,360) ;
\draw [color={rgb, 255:red, 0; green, 0; blue, 0 }  ,draw opacity=1 ][line width=0.75]    (20.38,165.53) -- (209,109) ;
\draw [color={rgb, 255:red, 0; green, 0; blue, 0 }  ,draw opacity=1 ][line width=0.75]    (20.38,165.53) -- (210,330) ;
\draw  [fill={rgb, 255:red, 0; green, 0; blue, 0 }  ,fill opacity=1 ] (414.4,218.3) .. controls (414.4,215.6) and (416.46,213.41) .. (419,213.41) .. controls (421.54,213.41) and (423.6,215.6) .. (423.6,218.3) .. controls (423.6,221) and (421.54,223.19) .. (419,223.19) .. controls (416.46,223.19) and (414.4,221) .. (414.4,218.3) -- cycle ;

\draw (189.58,291.99) node [anchor=north west][inner sep=0.75pt]    {$b_{\ell }$};
\draw (185.58,128.99) node [anchor=north west][inner sep=0.75pt]    {$b_{2}$};
\draw (256.5,152.24) node  [font=\large,color={rgb, 255:red, 0; green, 0; blue, 0 }  ,opacity=1 ]  {$u_{j}$};
\draw (157.81,189.73) node  [font=\small,color={rgb, 255:red, 208; green, 2; blue, 27 }  ,opacity=1 ,rotate=-14.7]  {$\{\mathbf{v} ,u\} \ \in E'?$};
\draw (185.58,34.99) node [anchor=north west][inner sep=0.75pt]    {$b_{1}$};
\draw (226.85,273.22) node  [font=\LARGE,rotate=-88.96] [align=left] {...};
\draw (206.5,233.24) node  [font=\large,color={rgb, 255:red, 0; green, 0; blue, 0 }  ,opacity=1 ]  {$u$};
\draw (342.73,205.53) node  [font=\normalsize]  {$s_{u}$};
\draw (186.01,5) node [anchor=north west][inner sep=0.75pt]  [font=\footnotesize]  {$\left( size\leq \sqrt{n}\right)$};
\draw (253,330) node [anchor=north west][inner sep=0.75pt]  [font=\footnotesize]  {$\ell \ =\frac{|N( v) |}{\sqrt{n}}$};
\draw (17.81,147.73) node  [font=\large,color={rgb, 255:red, 0; green, 0; blue, 0 }  ,opacity=1 ]  {$v$};

\end{tikzpicture}}}
\caption{Illustration for the local construction of the \emph{3-spanner}. On query $\{v, u\}$, namely when querying whether the edge $\{v, u\}$ belongs to the spanner, the algorithm returns NO since the vertex $u_j$ is in the same bucket as $u$ and $\{v,u_j\}$ is the edge that has minimum rank amongst the edges in the bucket that are incident to the cluster of $u$.}
\label{fig:3spanner}
\end{figure}
\fi

\subsection{Algorithm for constructing \texorpdfstring{$5$}{5}-spanners}\label{sec:5spanH}
We extend the ideas from the previous section to obtain our algorithm for constructing $5$-spanners as follows.
We partition the vertices in the graph into three sets: {\em heavy}, {\em medium}, and {\em light}.
The set of light vertices is defined to be the set of all vertices with a degree at most $n^{1/3}$ and the set of heavy vertices is defined to be the set of all vertices of degree at least $n^{2/3}$.
The set of the medium vertices is defined to be all vertices that are not light nor heavy.

As before, we add to the spanner all the edges incident to light vertices and cluster all the heavy vertices into a cluster of diameter $2$. 
The difference is that now when we partition the heavy vertices into sets according to their degrees the first set consists of all vertices with a degree in $[n^{2/3}+1, 2n^{2/3}]$.

We partition the set of medium vertices into two sets according to the following random process.
Each medium vertex $v$ samples uniformly at random $\Theta(\log n)$ of its neighbors. 
If one of the vertices in the sample is heavy then $v$ joins the cluster of the heavy vertex in the sample that has minimum rank. 
Otherwise we say that $v$ is {\em bad}.
This forms clusters of diameter $4$.

In a similar manner to the process described above we define another a new collection of sets of centers for the bad vertices such that the total number of such centers is $\tO(n^{2/3})$ and each bad vertex belongs to at least one cluster and at most $O(\log n)$ clusters. The new centers are selected (randomly) only from the set of vertices which are not heavy.  
We call the corresponding clusters {\em light-clusters} since they contain at most $n^{2/3}$ vertices and have diameter $2$.
Since the total number of light-clusters is $\tO(n^{2/3})$ we can afford to take an edge between every pair of adjacent light-clusters.
Moreover, we partition each light-cluster into buckets of size $n^{1/3}$ and take an edge between every pair of adjacent buckets. Since each bad vertex belongs to $O(\log n)$ light-clusters, the total number of pairs of buckets is $\tO(n^{4/3})$.
The time and probe complexity of finding all the edges incident to two buckets is $\tO(n^{2/3})$, as desired.

To analyse the stretch factor we partition the edges we remove into three types. The first type of edges are edges between vertices in the same cluster. The second type of edges are edges between a vertex $v$ and a cluster $C$ which is not light, in which case there at least one edge in the spanner which is incident to both $v$ and $C$. The third type of edges are edges which are incident to a pair of light-clusters, in which case there exists at least one edge in the spanner which is incident to each pair of such clusters.
Thus, overall the stretch factor is $5$, as claimed.
See Figures~\ref{fig:5-light} and~\ref{fig:5spanner_rep} for illustrations of this part. 

\subsection{Algorithm for constructing \texorpdfstring{$O(k^2)$}{O(k\^2)}-spanners}\label{sec:kspanH}

The high-level idea of the algorithm for constructing \(O(k^2)\)-spanners, which we describe from a global point of view, follows that of \cite{PRVY19}.
The vertices of the graph are first partitioned into $\tO(n^{2/3})$ Voronoi cells which are formed with respect to a randomly selected set of $\tO(n^{2/3})$ centers.
We can assume that each Voronoi cell has diameter $O(k)$ by using a separate algorithm to handle {\em remote} vertices which may not be close to a center.
Each Voronoi cell is then partitioned into clusters of size $L = \tO(n^{1/3})$. In addition each Voronoi cell is marked with probability $1/n^{1/3}$ which respectively also marks all the clusters of the cell.
Each non-marked cluster connects to all the adjacent marked clusters using a single edge. This forms {\em clusters-of-clusters} around marked clusters.
Instead of connecting every pair of adjacent clusters $A$ and $B$, which we can not afford, our goal is to connect $A$ with the cluster-of-clusters of some marked cluster $C$ adjacent to $B$.
Since we can not afford to reconstruct the cluster-of-clusters of $C$,
we instead find the identity of all the Voronoi cells which are adjacent to $C$ and try to connect $A$ with at least one of these cells.
We show that this is indeed the case although $A$ may not be connected directly to any one of these cells.
By applying an inductive argument we show that the number of hops between $A$ and $B$ is $O(k)$, where traversing from one Voronoi cell to another is considered one hop.
Since the diameter of each Voronoi cell is $O(k)$ we obtain an overall stretch factor of $O(k^2)$.

We improve on the LCA of \cite{PRVY19} in two main ways.
The first main improvement is a new method for partitioning the Voronoi cells into clusters of size $L$, allowing the cluster containing a vertex to be reconstructed using \(O(\Delta^2L^2)\) probes instead of \(O(\Delta^3L^2)\).
The second improvement relates to the problem of connecting a cluster $A$ to the cluster-of-clusters of some marked cluster $C$ adjacent to $B$.
In particular there is an issue of which marked cluster $C$ should be chosen, since it is too expensive to reconstruct every marked cluster adjacent to $B$.
The LCA of \cite{PRVY19} processes a single cluster of each adjacent marked Voronoi cell to $B$,
of which there may be as many as \(\tO(\Delta)\).
We instead devise a rule by which $B$ is {\em engaged} with a single marked cluster adjacent to it,
and show that in fact it suffices to only consider this one cluster.
Combining these improvements reduces the total number of probes from \(\tO(n^{2/3}\Delta^4)\) to \(O(n^{2/3}\Delta^2)\).

\subsection{Algorithm for graph decomposition}\label{sec:decomp_intro}
To further reduce the runtime of Theorem~\ref{thm:k2-main}, we develop a new local computation algorithm to decompose graphs into subgraphs with smaller degree. Observe that for a graph $G$, for subgraphs $G_1,\ldots,G_t$, if we have $k$-spanners $H_i\subset G_i$ for every $i$, the union $\bigcup_{i\in [t]}H_i$ is a $k$-spanner for $G$. As the runtime of Theorem~\ref{thm:k2-main} depends on the maximum degree $\Delta$, we develop an efficient LCA to break $G$ into $t$ graphs, each with maximum degree $O(\max\{\Delta/t,\log n\})$, where $t$ is a parameter to be chosen. Given $v$ and an index $i\in [t]$, the LCA returns in time $O(\Delta/\sqrt{t})$ all neighbors of $v$ in $G_i$ (i.e. it supports $\allnbr$ queries to each subgraph). We believe this algorithm may have other applications.

To apply this algorithm in the spanner framework, we compose the LCA for $O(k^2)$ spanners with the LCA for graph decomposition. This is more subtle than generic sequential composition of algorithms, as we must ensure the per-query overhead is mild. We do this by observing the $O(k^2)$-spanner algorithm only ever makes all neighbor queries, and so the decomposition LCA spends $O(\Delta/\sqrt{t})$ work per query the spanner LCA makes to the graph. In particular, as the spanner LCA makes $O(n^{2/3}\Delta)$ all neighbor queries, our new runtime is $O(n^{2/3}\Delta^2/t^{3/2})$ given our choice of $t$. As decomposing $G$ into $t$ subgraphs increases the size of the output spanner by a factor of $t$, we ultimately balance parameters and obtain a probe and time complexity of $O(n^{2/3-(1.5-\alpha)/k}\Delta^2)$ for any $\alpha > 0$.

\subsubsection{Decomposing the Graph}
To build this graph, consider assigning each edge of $G=(V,E)$ two colors $i,j\in [R]$, one from each endpoint. In particular, $v$ assigns its first $\Delta/R$ edges to receive color $1$, the next $\Delta/R$ to receive color $2$, etc. Then if an edge $(u,v)$ has received colors $i,j$ from both endpoints, we can let the overall edge color be $(i,j)$ where we assume w.l.o.g that $u<v$. Observe that if each color corresponds to a subgraph, then this decomposition breaks $G$ into $R^2$ subgraphs. Moreover, given $i,j$ and a vertex $v$, we can quickly enumerate blocks $i$ and $j$ from $v$ and determine which edges lie in the specified subgraph. However, as these blocks may be poorly aligned, this as described results in a maximum subgraph degree of $\Delta/R$ rather than $\Delta/R^2$. Instead, we have each vertex choose a random shift, and assign labels to its blocks according to this shift. Then via standard concentration bounds the maximum degree of every subgraph is as desired. We remark that as we must enumerate every element of bucket $i$ and $j$ from $v$ to find edges with label $(i,j)$, the worst-case time for an individual neighbor query can be up to $\Omega(\Delta/R)$. However, we only need to do this once to answer an all-neighbors query, so as long as all the neighbors are desired we can efficiently amortize this cost.

\subsection{The number of random bits}

All our algorithms are randomized and hence use random bits. For results \ref{item1}-\ref{item2} 
we use randomness in the selection of centers and representatives. In result \ref{item3} we use randomness in the selection of centers, marked clusters, and random ranking of edges. 

When the centers and representatives are selected independently, the arguments for proving the guarantees on the sparsity of the spanner follow from standard concentration bounds.
As shown by Parter at al.~\cite{PRVY19}, by using a less standard analysis one is able to prove that the same guarantees on the sparsity hold even when the random bits are only $\Theta(\log n)$-wise independent (which requires only $O(\log^2 n)$ truly random bits).
Furthermore, by using an intricate analysis, they showed that the guarantees on the stretch factor continue to hold as well.  
In this writing, we do not repeat the analysis in~\cite{PRVY19} since it lends itself quite easily to our setting. 

For result~\ref{item4}, similar techniques to those of~\cite{PRVY19} allow the result to be implemented using $\polylog(n)$-wise independence as well~\footnote{More specifically, by using the concentration bound from Fact 5.3 in~\cite{PRVY19} on the sum of $d$-wise independent random variables.}.

\subsection{Related work}
\label{subsec:related}
As mentioned above, the work which is the most closely related to our work is by Parter et el. \cite{PRVY19}.
In addition to the upper bounds mentioned in Section~\ref{sec:our} they also observe that it is possible to obtain an LCA for constructing $5$-spanners with $\tO(n^{1+1/k})$ edges and probe complexity $\tO(n^{1-1/(2k)})$ for the special case in which the minimum degree is known to be at least $n^{1/2-1/(2k)}$ (this builds on the fact that by picking $\tO(n^{(1+1/k)/2})$ centers, w.h.p. each vertex has a center in its neighborhood).
In addition to upper bounds, they also provide a lower bound of $\Omega(\min\{\sqrt{n},\frac{n^2}{m}\})$ probes for the simpler task of constructing a spanning graph with $o(m)$ edges, where $m$ denotes the number of edges in the input graph.

Our work also builds on the upper bound in~\cite{LL18}, designed originally for bounded degree graphs, which provide a spanner with $(1+\epsilon) n$ edges on expectation, where $\eps$ is a parameter, stretch factor $O(\log^2n \cdot \poly(\Delta/\epsilon))$ and probe complexity of $O(\poly(\Delta/\epsilon) \cdot n^{2/3})$.
The work in~\cite{LL18} is a follow-up of~\cite{LRR20,LRR16} which initiated the study of LCAs for constructing ultra-sparse (namely, with $(1+\epsilon) n$ edges) spanning subgraphs.

\section{Preliminaries}
\label{sec:preliminaries}
The input graph $G=(V,E)$ is a simple undirected graph with $|V|=n$ vertices and a bound on the degree $\Delta$.
Both parameters $n$ and $\Delta$ are known to the algorithm.
Each vertex $v\in V$ is represented as a unique ID from $[n]$.

A local algorithm has access to the \emph{adjacency list oracle} $\mathcal{O}^G$ which provides answers to the following probes (in a single step):
\begin{itemize}
\item \textbf{Degree probe:} Given $v\in V$, returns the degree of $v$, denoted by $\deg(v)$.  
\item \textbf{Neighbour probe:} Given $v \in V$  and an index $i$, returns the $i^\textrm{th}$ neighbor of $v$ if $i \leq \deg(v)$. Otherwise, $\bot$ is returned.
Additionally, for $v\in V$, we define the all-neighbors query, denoted by $\allnbr(v)$, which returns all the neighbors of $v$. Clearly, this query can be implemented by $\deg(v) + 1$ neighbor probes.  
\item \textbf{Adjacency probe:} Given an ordered pair $\langle u,v \rangle$ where $u\in V$ and $v \in V$, if $v$ is a neighbor of $u$ then $i$ is returned where $v$ is the $i^\textrm{th}$ neighbor of $u$. Otherwise, $\bot$ is returned.
\end{itemize}

We denote the distance between two vertices $u$ and $v$ in $G$ by $\dist(u,v)$ and the set of neighbours of $v$ in $G$ by $N_G(v)$.
We denote by $N_G(v)[i]$ the $i$-th neighbour of $v$ in $G$.
For vertex $v \in V$ and an integer $k$,
let $\Gamma_k(v,G)$ denote the set of vertices at distance at most
$k$ from $v$.
When the graph $G$ is clear from the context, we shall use the shorthand
$d(u,v)$, $N(v)$ and $\Gamma_k(v)$  for $d_G(u,v)$, $N_G(v)$ and $\Gamma_k(v,G)$, respectively.
\sloppy
We define a ranking $r$  of the edges as follows:
$r(u,v) < r(u',v')$ if and only if $\min\{u, v\} < \min\{u', v'\}$
or $\min\{u, v\} = \min\{u',v'\}$
and $\max\{u,v\} < \max\{u',v'\}$.

\medskip
We shall use the following definitions in our algorithms for constructing $3$-spanners and $5$-spanners. 

\begin{definition}\label{def:class}
We say that a vertex $v\in V$ is in {\em class $i \in \mathbb{N}$ w.r.t. $\Delta$} if $\deg(v) \in [2^{i-1}\Delta+1, 2^i\Delta]$.
 \end{definition}

\begin{definition}\label{def:bucket}
We say that an index $i \in \mathbb{N}$ is in {\em bucket $j\in \mathbb{N}$ w.r.t. $\Delta$} if $i \in [(j-1)\cdot\Delta+1, j\cdot \Delta]$.
 \end{definition}

\subsection{Probes in the LCA model}\label{subsec:model}
Since the introduction of the model in~\cite{RTVX11}, there have been several formulations concerning, mainly, the measure of performance, the way the input is accessed, and whether preprocessing is allowed.   
In particular, when the input is a graph, there is the question of whether the LCA can probe the graph anywhere (i.e. ask for the neighbors of an arbitrary vertex).
In contrast to message-passing models such as {\rm CONGEST} and distributed {\rm LOCAL} algorithms, in LCAs the standard assumption is that indeed the LCA can access the graph anywhere and more specifically that each vertex in the input graph is represented as a unique ID from $[n] = \{1,\ldots,n\}$.
To support this claim, we refer the reader to the ultra-formal definition  in~\cite{Gol17} (Definition 12.11) as well as~\cite{lec, EMR18}.

We note that the utility of making far-probes~\footnote{Namely, probing vertices for which we do not yet know a path from the query vertex.} was studied in~\cite{GHLMS16}, in which the authors showed that for a large family of problems, this extra power is not so useful. 
Indeed, this extra power is not always used by LCAs. For example, in the recent result of Ghaffari~\cite{Gha22}, which provides an LCA for the problem of Maximal Independent Set, the assumption is that the IDs are taken from $[n^{10}]$. 
Nonetheless, we stress that this extra power is an important feature of the LCA model, which, in particular, distinguishes it from message-passing models (see more on the difference between LCAs and distributed {\rm LOCAL} algorithms in Section 4.1 in~\cite{LM17}) and comes into play in problems which have a more global nature.
For example, this extra power is utilized in Prop. 12.13 in~\cite{Gol17} for graph coloring and in~\cite{LRY17} for approximate Maximum-Matching. The latter LCA is used in Behnezhad et al.~\cite{BRR22} to obtain a state-of-the-art sublinear algorithm for the extensively studied problem of approximate Maximum-Matching.

\section{LCA for constructing \texorpdfstring{$3$}{3}-spanners}\label{sec:3spanner}
In this section, we prove the following theorem. Due to space limitations, we defer the claims regarding probe and time complexities as well as the stretch factor to the appendix. 

\begin{theorem}\label{thm:3-main}
There exists an LCA that given access to an $n$-vertex simple undirected graph $G$, constructs a $3$-spanner of $G$ with $\tO(n^{1+1/2})$ edges whose probe complexity and time complexity are $\tO(n^{1/2})$.
\end{theorem}

Our algorithm is listed as Algorithm~\ref{alg:spanner3}. As mentioned-above our algorithm proceeds by forming clusters around centers and connecting the different clusters.
To make the description of our algorithm complete we begin with describing the selection of centers.
 
 We define $t \eqdef \log \sqrt{n}$ sets of centers $S_1, \ldots, S_t$. 
   For every $i\in [t]$, we pick u.a.r. $x_i$ vertices to be in $S_i$ where
$x_1 = \sqrt{n}\log n$ and $x_{i+1} = x_i/2$ for every $i \in [t-1]$. The rest of the details of the algorithm appear in Algorithm~\ref{alg:spanner3}. We next prove the correctness of the algorithm.

Recall that we refer to a vertex whose degree is greater than $\sqrt{n}$ as heavy. 
The next claim states that with high probability every heavy vertex has at least one center and $O(\log n)$ centers in its neighborhood.  
 
\begin{claim}\label{clm:numberc}
With high probability, for every $i\in [t]$ and every vertex $v\in V$ that is in class $i$ w.r.t. $\sqrt{n}$ it holds that $N(v) \cap S_i \neq \emptyset$ and that $|N(v) \cap S_i| = O(\log n)$.
\end{claim}

\begin{algorithm}
\caption{LCA for constructing $3$-spanners
}\label{alg:spanner3}
\textbf{Input:} Access to an undirected graph $G=(V,E)$ and a query $\{u, v\}\in E$ where we assume w.l.o.g. that $\deg(u) \geq \deg(v)$.\\
\textbf{Output:} Returns whether $\{u, v\}$ belongs to the spanner or not.
\begin{enumerate}
\item\label{alg:3:step1} If $\deg(v) \leq n^{1/2}$ then return YES (recall that $\deg(u) \geq \deg(v)$).
\item\label{alg:3:step2} Otherwise, let $c$ denote the class of $u$ w.r.t. $\sqrt{n}$ (see Definition~\ref{def:class}). 
\item\label{alg:3:step3} If $v \in S_c$ then return YES.
\item\label{alg:3:step4} Otherwise, let $\mathcal{C} \eqdef S_c \cap N(u)$. If $\cal{C} = \emptyset$ then return YES.
\item\label{alg:3:step5} Let $i$ denote the index of $u$ in $N(v)$ and let $b$ denote the {\em bucket} of $i$ w.r.t. $\sqrt{n}$ (see Definition~\ref{def:bucket}).
\item\label{alg:3:step6} For each $x \in \cal{C}$:
\begin{enumerate}
\item\label{alg:3:step6a} Go over every $j < i$ such that $j$ is in bucket $b$ and return YES if for every such $j$, $N(v)[j]$ does not belong to the cluster of $x$. 
\end{enumerate}
\item Return NO.
\end{enumerate}
\end{algorithm}

\def\claims1{
\begin{claim}
With high probability, the stretch factor of the spanner constructed by Algorithm~\ref{alg:spanner3} is $3$.
\end{claim}
\begin{proof}
Let $\{u, v\}$ be an edge in $E$ such that $\deg(u) \geq \deg(v)$.
We will show that there exists a path of length at most $3$ between $u$ and $v$ in the the spanner constructed by Algorithm~\ref{alg:spanner3} denoted by $G' = (V, E')$.
If $\deg(v) \leq \sqrt{n}$ then $\{u, v\} \in E'$ and we are done.
Otherwise, if there exists a cluster $C$ such that $u$ and $v$ are both belong to $C$ then in $G'$ they are both connected by an edge to the center of $C$. Thus there exists a path of length at most $2$ between $u$ and $v$ in $G'$. 
Otherwise, let $C'$ be a cluster for which $u$ belongs to. We claim that $v$ is adjacent to $C'$ in $G'$. This follows by induction on the index of $u$ in $N(v)$ and Sub-Step~\ref{alg:3:step6a}. 
\end{proof}

\begin{claim}
The probe and time complexity of  Algorithm~\ref{alg:spanner3} is $O(\sqrt{n}\log n)$.
\end{claim}
\begin{proof}
Steps~\ref{alg:3:step1}-\ref{alg:3:step2} can be implemented by making a single degree probe. Their time complexity is $O(1)$.
Step~\ref{alg:3:step3} can be implemented by accessing the random coins.
To implement Step~\ref{alg:3:step4} we need to go over all the vertices in $S_c$ (we may assume w.l.o.g. that we generate all the centers in advance as there are only $O(\sqrt{n}\log n)$ centers) and check whether they are in $N(u)$ (by making a single adjacency probe). 
Thus the probe (and time) complexity of this step is $O(\sqrt{n}\log n)$. 
Step~\ref{alg:3:step5} can be implemented by a single adjacency probe.
The total number of vertices we check in Sub-Step~\ref{alg:3:step6a} is bounded by the size of $\mathcal{C}$ times the size of a bucket which is $\sqrt{n}$. For each vertex we check we make a single neighbor and then we check whether it belongs to the cluster of a specific center. The latter can be implemented by making a single degree probe and a single adjacency probe. 
By Claim~\ref{clm:numberc}, the size of $\mathcal{C}$ is bounded by $O(\log n)$, thus the probe (and time) complexity of Step~\ref{alg:3:step6} is $O(\sqrt{n} \log n)$. The claim follows.
\end{proof}
}

\ifnum\confver=0
\claims1
\fi

\begin{claim}
With high probability, the number of edges of the spanner constructed by Algorithm~\ref{alg:spanner3} is $\tO(n^{1+1/2})$.
\end{claim}

\begin{proof}
The number of edges added to $E'$ due to Step~\ref{alg:3:step1} is at most $n^{3/2}$.
By the bound on the number of centers, the number of edges added to $E'$ due to Step~\ref{alg:3:step3} is $O(n^{3/2}\log n)$.
To analyse the number of edges added to $E'$ due to Step~\ref{alg:3:step6} consider an edge $\{u, v\}$ such that $\deg(u) \geq \deg(v)$, $\deg(v) > \sqrt{n}$ and $v\notin S_c$, where $c$ denotes the class of $u$ w.r.t. $\sqrt{n}$.
Since $u$ is in class $c$ it follows that $\deg(u) \leq 2^{c}\sqrt{n}$. Since $\deg(v) \leq \deg(u)$ it follows that $N(v)$ has at most $2^{c}$ buckets. 
By Sub-Step~\ref{alg:3:step6a}, for any cluster $C$, the number of edges in $E'$ that are incident to $v$ and a vertex from $C$ is at most $2^{c}$ (since we add to $E'$ at most a single edge for each bucket of $N(v)$). Since the number of clusters of class $c$ is $O(\sqrt{n}\log n/2^c)$, the total number of clusters of class greater or equal to $c$ is $O(\sqrt{n}\log n/2^c)$ as well.
Therefore, the total number of edges that are incident to $v$ and added to $E'$ due to Step~\ref{alg:3:step6} is $O(\sqrt{n}\log n)$.
By Claim~\ref{clm:numberc}, w.h.p., the number of edges that are added due to Step~\ref{alg:3:step4} is $0$. 
We conclude that the $|E'| = O(n^{3/2}\log n)$, as desired.
\end{proof}

\section{LCA for constructing \texorpdfstring{$5$}{5}-spanners}\label{sec:5spanner}
In this section, we prove the following theorem.
\begin{theorem}\label{thm:5-main}
There exists an LCA that given access to an $n$-vertex simple undirected graph $G$, constructs a $5$-spanner of $G$ with $\tO(n^{1+1/3})$ edges whose probe complexity and time complexity are $\tO(n^{2/3})$.
\end{theorem}

Our algorithm for constructing $5$-spanners also proceeds by forming clusters around centers and connecting the different clusters. 
For the sake of presentation, we first describe our local algorithm from a global point of view (see algorithm~\ref{algGlobal:spanner5}). In Section~\ref{sec:localimp5} we describe the local implementation of this algorithm.

As in the algorithm for constructing $3$-spanners, the clusters are formed around randomly selected centers only that now we have two types of clusters (and centers), {\em heavy-clusters} and {\em light-clusters} that will be described in the sequel.

\paragraph*{The selection of the first type of centers.}
The selection of the first type of centers proceeds as follows. We define $a \eqdef \log n^{1/3}$ sets of centers $S^1_1, \ldots, S^1_a$. 
   For every $i\in [a]$, we pick u.a.r. $y_i$ vertices to be in $S^1_i$ where 
$y_1 = n^{1/3}\log n$ and $y_{i+1} = y_i/2$ for every $i \in [a-1]$. 
The clusters which are formed around the first type of centers are the {\em heavy-clusters}. The formation of the heavy-clusters is described in Step~\ref{algGlobal5:step2} of Algorithm~\ref{algGlobal:spanner5}.

\paragraph*{The selection of the second type of centers.}
The selection of the second type of centers proceeds as follows. We define $b \eqdef \log n^{1/3}$ sets of centers $S^2_1, \ldots, S^2_b$. 
   For every $i\in [b]$, we pick u.a.r. $x_i$ vertices to be in $S^2_i$ where 
$x_1 = n^{2/3}\log n$ and $x_{i+1} = x_i/2$ for every $i \in [b-1]$. 

The clusters which are formed around the second type of centers are the {\em light-clusters}. The formation of the light-clusters is described in Step~\ref{algGlobal5:step3} of Algorithm~\ref{algGlobal:spanner5}.

\medskip\noindent
The way we connect the different clusters is described in Steps~\ref{algGlobal5:step4} and~\ref{algGlobal5:step5}.

\begin{algorithm}
\caption{Global algorithm for constructing $5$-spanners \label{algGlobal:spanner5}}
\textbf{Input:} A graph $G=(V,E)$.\\
\textbf{Output:} Constructs a $5$-spanner of $G$, $G'=(V, E')$.
\begin{enumerate}
\item\label{algglobstep1}\label{step:span5-1} For every $v$ such that $\deg(v) \leq n^{1/3}$ add to $E'$ all the edges that are incident to $v$.
\item\label{algGlobal5:step2} Forming heavy-clusters:
\begin{enumerate}
    \item\label{step:span5-2a} For each vertex $v$ such that $\deg(v) \geq n^{2/3}$ we define the centers of $v$ to be $N(v) \cap S^1_c$ where $c$ is the class of $v$ w.r.t. $n^{2/3}$ (see Definition~\ref{def:class}).
    For every center $s$ of $v$, $v$ joins the cluster of $s$ by adding the edge $\{s, v\}$ to $E'$.
    \item\label{algGlobal5:step2b}\label{step:span5-2b} Each vertex $v$ such that $n^{1/3} < \deg(v) < n^{2/3}$ sample u.a.r. $y \eqdef \Theta(\log n)$ of its neighbors. Let $R_v$ denote this set. The {\em representative} of $v$ is defined to be the vertex, $r$, of minimum id in $R_v$ such that $\deg(r) \geq n^{2/3}$ (if such vertex exists).
    If $v$ has a representative, $r$, then the edge $\{v, r\}$ is added to $E'$ (and hence $v$ joins all the clusters of $r$). 
    See Figure~\ref{fig:5spanner_rep} for illustration of forming heavy-clusters with {\em representative}.

\end{enumerate}
\item\label{algGlobal5:step3} Forming light-clusters:
\begin{enumerate}
    \item\label{step:span5-3a} For each vertex $v$ such that $n^{1/3} < \deg(v) < n^{2/3}$ for which $v$ does not have a representative we define the centers of $v$ to be $N(v) \cap S^2_c$ where $c$ is the class of $v$ w.r.t.  $n^{1/3}$ (see Definition~\ref{def:class}).
    For every center $s$ of $v$, $v$ joins the cluster of $s$ by adding the edge $\{s, v\}$ to $E'$.
\end{enumerate}
\item\label{algGlobal5:step4} Connecting vertices to adjacent heavy-clusters:
\begin{enumerate}
    \item Let $\{u, v\}$ be such that $\deg(u) \geq \deg(v)$ and $u$ belongs to a heavy-cluster. For each cluster $C$ that $u$ belongs to, do:
    \begin{enumerate}
    \item Partition the interval $[\deg(v)]$ into sequential intervals, which we refer to as buckets, of size $n^{2/3}$: $b_1, \ldots, b_s$ (where only $b_s$ may have size which is smaller than $n^{2/3}$). 
    \item\label{step4c}\label{step:span5-4c} For each $i\in [s]$, go over every $j \in b_i$ in increasing order and check if $N(v)[j]$ belongs to $C$. If such $j$ is found, add $\{v, N(v)[j]\}$ to $E'$ and move to the next bucket.
    \end{enumerate}
\end{enumerate}
\item\label{algGlobal5:step5} Connecting adjacent light-clusters:

\begin{enumerate}
    \item Let $\{u, v\}$ be such that both $u$ and $v$ belong to different light-clusters. 
    For each light clusters $C_u$ and $C_v$ that $u$ and $v$ belong to, respectively, do:
    \begin{enumerate}
    \item Let $s_u$ and $s_v$ denote the centers of $C_u$ and $C_v$, respectively. Let $c_u$ and $c_v$ denote the classes of $u$ and $v$ w.r.t. $n^{1/3}$, respectively.
    \item\label{step:prev} Partition the vertices in $N(s_u)$ that belong to the cluster $C_u$ (namely, the neighbors of $s_u$ that are in class $c_u$ w.r.t. $n^{1/3}$) into subsets of size $n^{1/3}$ greedily by their index in $N(s_u)$, $S^u_1, \ldots, S^u_t$ (all the subsets are of size $n^{1/3}$ except from perhaps $S^u_t$).
    \item\label{step:prev2} Repeat Step~\ref{step:prev} for the vertices in $N(s_v)$ that belong to $C_v$ and let $S^v_1, \ldots, S^v_r$ denote the resulting subsets.
    \item\label{stepminrank}\label{step:span5-5d} For each $i\in [t]$ and $j \in [r]$, add the edge of minimum rank in $E(S^u_i, S^v_j)$ to $E'$ (if such edge exists).
    \end{enumerate}
    See Figure~\ref{fig:5-light} for an illustration of connecting to adjacent light-clusters.

\end{enumerate}
\end{enumerate}

\end{algorithm}
In the next couple of claims we prove that with high probability every vertex $v$ such that $\deg(v) > n^{1/3}$ joins at least one cluster and at most $O(\log n)$ clusters. 
To do so, we partition the vertices with degree greater than into $n^{1/3}$ into $3$ sets.
The first set, denoted by $H$, is the set of vertices, $v$, such that $\deg(v) \geq n^{2/3}$.
The second set is the set of vertices, $v$, such that $n^{1/3} < \deg(v) < n^{2/3}$ for which at least half of the vertices in $N(v)$ have degree at least $n^{2/3}$. We denote this set by $M_1$.
$M_2$ consists of the remaining vertices. Namely, $M_2$ is the set of vertices, $v$, such that $n^{1/3} < \deg(v) < n^{2/3}$ and for which less than half of the vertices in $N(v)$ have degree at least $n^{2/3}$.

\medskip
The implication of the next claim is that w.h.p. every vertex in $H$ joins at least one heavy-cluster and at most $O(\log n)$ heavy-clusters.

\input{figures/fig5ll.tex}

\begin{claim}\label{clm:numberc2}
With high probability, for every $v\in H$ it holds that $N(v) \cap S^1_c \neq \emptyset$ and that $|N(v) \cap S^1_c| = O(\log n)$ where $c\in [a]$ is the class of $v$ w.r.t. $n^{2/3}$. 
\end{claim}

The implication of the next claim (when combined with Claim~\ref{clm:numberc2}) is that w.h.p. every vertex in $M_1$ joins, via a representative, at least one heavy-cluster and at most $O(\log n)$ heavy-clusters.
\begin{claim}\label{clm:numberc21}
With high probability, for every $v\in M_1$ it holds that $v$ has a representative. 
\end{claim}
\begin{proof}
Let $v \in M_1$. Consider Step~\ref{algGlobal5:step2b} of Algorithm~\ref{algGlobal:spanner5}. Since at least half of the neighbors of $v$ have degree at least $n^{2/3}$, it follows that w.h.p. $R_v \neq \emptyset$ and so $v$ has a representative. 
\end{proof}

The implication of the next claim is that w.h.p. every vertex in $M_2$ that does not have a representative joins at least one light-cluster and at most $O(\log n)$ light-clusters.
\begin{claim}\label{clm:numberc22}
With high probability, for every $v\in M_2$ it holds that $N(v) \cap S^2_c \neq \emptyset$ and that $|N(v) \cap S^2_c| = O(\log n)$ where $c\in [b]$ is the class of $v$ w.r.t. $n^{1/3}$. 
\end{claim}

The following corollary follows directly from Claims~\ref{clm:numberc2}-\ref{clm:numberc22}.
\begin{corollary}\label{cor1}
With high probability every vertex $v$ such that $\deg(v) > n^{1/3}$ joins at least one cluster and at most $O(\log n)$ clusters.
\end{corollary}

\begin{claim}
With high probability, $|E'| = \tO(n^{1+1/3})$.
\end{claim}

\begin{proof}
The number of edges added to $E'$ due to Step~\ref{step:span5-1} is at most $n^{1+1/3}$.
By Claims~\ref{clm:numberc2} and~\ref{clm:numberc22} the number of edges added to $E'$ due to Steps~\ref{step:span5-2a} and~\ref{step:span5-3a} is $\tO(n)$.
Since each vertex has at most one representative the number of edges added to $E'$ due to Step~\ref{step:span5-2b} is at most $n$.

Consider $\{u, v\}$ such that $\deg(u) \geq \deg(v)$ and $u$ belongs to a heavy-cluster $C$.
According to Step~\ref{step:span5-4c} we connect $v$ to $C$ by adding to $E'$ at most $\lceil\deg(v)/n^{2/3}\rceil$ edges (at most one edge for each bucket of $N(v)$).

If $\deg(u) \leq n^{2/3}$ then $\deg(v) \leq n^{2/3}$ as well and so $\lceil\deg(v)/n^{2/3}\rceil \leq 1$. Therefore the total number of edges that are incident to $v$ and added to $E'$ due to Step~\ref{step:span5-4c} is bounded by the total number of centers of the first type which is $\tO(n^{1/3})$.

Otherwise, let $c$ denote the class of $u$ w.r.t. $n^{2/3}$, then by definition $\deg(u) \leq 2^c \cdot n^{2/3}$. Therefore by our assumption $\deg(v) \leq 2^c \cdot n^{2/3}$ as well.
The number of centers in $S^1_c$ is $n^{1/3} \log n / 2^c$ and so the total number of centers in $\bigcup_{c \leq i \leq a} S^1_i$ is $O(n^{1/3} \log n / 2^c)$.
Observe that the number of edges which are incident to $v$ and added to $E'$ due to Step~\ref{step:span5-4c} is at most $\lceil\deg(v)/n^{2/3}\rceil$ times the number of centers in $\bigcup_{c \leq i \leq a} S^1_i$.
Thus the total number of edges that are incident to $v$ and added to $E'$ due to Step~\ref{step:span5-4c} is $\tO(n^{1/3})$ in this case as well.
Therefore, the total number of edges which are added to $E'$ in Step~\ref{step:span5-4c} is $\tO(n^{1+ 1/3})$.

By Claim~\ref{clm:numberc22} it follows that the total number of subsets partitioning the light clusters is $\tO(n^{2/3})$ as the size of each subset is $n^{1/3}$ except for at most $\tO(n^{2/3})$ subsets, and since each vertex may belong to $O(\log n)$ different clusters.
Since in Step~\ref{step:span5-5d} we add at most a single edge between a pair of subsets the total number of edges added to $E'$ due to this step is $\tO(n^{1 + 1/3})$.
This concludes the proof of the claim.

\end{proof}

\begin{claim}
With high probability, the stretch factor of the spanner constructed by Algorithm~\ref{algGlobal:spanner5} is $5$.
\end{claim}

\begin{proof}
Let $\{u, v\}$ be an edge which is not included in $E'$.
By Step~\ref{algglobstep1} of the  algorithm it follows that the degree of both $u$ and $v$ is greater than $n^{1/3}$.
By Corollary~\ref{cor1} w.h.p. all vertices with degree greater than $n^{1/3}$ join at least one cluster. 
In the rest of the proof we condition on the event that both $u$ and $v$ join at least one cluster.

Assume w.l.o.g. that $\deg(u) \geq \deg(v)$.
If both $u$ and $v$ belong to the same cluster (either heavy or light) then there exists a path of length at most $4$ in $G'$ between $u$ and $v$ as the diameter of each cluster is at most $4$.

If $u$ belongs to a heavy cluster, $C$, then by Step~\ref{step4c} of the algorithm it follows that there exists at least one edge in $E'$ which is incident to $v$ and a vertex in $C$.
Since the diameter of $C$ is at most $4$ it follows that there exists a path in $G'$ from $v$ to $u$.

Otherwise, both $u$ and $v$ belong to different light clusters $C_u$ and $C_v$. By Step~\ref{stepminrank}, there exists at least one edge in $E'$ which is incident to a vertex in $C_u$ and a vertex in $C_v$. Since the diameter of a light cluster is at most $2$ we obtain that there exists a path in $G'$ from $u$ to $v$ of length at most $5$.
This concludes the proof of the claim.
\end{proof}

\def\localimp5{
\subsection{The local implementation}\label{sec:localimp5}

In this section we prove the following claim. In the proof of the claim we also describe the local implementation of Algorithm~\ref{algGlobal:spanner5}.

\begin{claim}
The probe and time complexity of the local implementation of Algorithm~\ref{algGlobal:spanner5} is $\tO(n^{2/3}\log n)$.
\end{claim}

\begin{proof}
On query $\{u, v\}$ we first probe the degree of $u$ and $v$ and return YES if either $u$ or $v$ have degree which is at most $n^{1/3}$.
Otherwise, assume w.l.o.g. that $\deg(u) \geq \deg(v)$. we consider the following cases. 
\paragraph*{First case: $\deg(u) \geq n^{2/3}$.}
In this case we find the centers of $u$ by going over all the centers, $s$, in $S^1_c$ where $c$ is the class of $u$ w.r.t. $n^{2/3}$ and preforming the adjacency probe $\langle u,s \rangle$.
If $v$ belongs to the set of centers of $u$ then we return YES.
Overall, since the number of centers of the first type is $\tO(n^{1/3})$, finding the centers of $u$ requires $\tO(n^{1/3})$ probes and time.

We then find the bucket $b$ of $u$ in $N(v)$ w.r.t. $n^{2/3}$ (see Definition~\ref{def:bucket}) by preforming the adjacency probe $\langle v,u \rangle$.
Let $i$ denote the index of $u$ in $N(v)$.
For each center of $u$, $s$ and for each $j\in b$ such that $j < i$, we check if $N(v)[j]$ belongs to the cluster of $s$. In order to do so we first probe the degree of $y = N(v)[j]$. If $\deg(y) \geq n^{2/3}$ then $v$ is in the cluster of $s$ if and only if it is a neighbour of $s$ and is in class $c$ w.r.t. $n^{2/3}$ where $c$ is such that $s$ belongs to $S^1_c$. 
If $\deg(y) < n^{2/3}$ then we first find the representative of $y$ and if it has a representative we check if it belongs to the cluster of $s$.
Since we have to check this for at most $n^{2/3}$ vertices and for $O(\log n)$ centers, overall the probe and time complexity of preforming this task is $\tO(n^{2/3})$.   

\paragraph*{Second case: $n^{1/3} < \deg(u) < n^{2/3}$ and either $u$ or $v$ have a representative.}
In this case we proceed as in the previous case only that we preform all the checks with respect to the centers of the representative of $u$ (and/or the representative of $v$). Since finding the representative of a vertex requires $O(\log n)$ probes and time the probe and time complexity in this case is $\tO(n^{2/3})$ as well.

\paragraph*{Third case: $n^{1/3} < \deg(u) < n^{2/3}$ and both $u$ and $v$ do not have representatives.}

This corresponds to the case in which both $u$ and $v$ belong to light clusters. In order to find the centers of $u$ we simply go over all vertices, $y$, in $N(u)$ and check if $y$ is in $S^2_c$ where $c$ denotes the class of $u$ w.r.t. $n^{1/3}$. We repeat the same process for $v$. Since checking if a vertex belongs to $S^2_c$ can be done in $O(\log n)$ time (we can generate all the centers in advance and store them in a binary search tree) this task requires $\tO(n^{2/3})$ probes and time.

Finally, for each pair of centers $s_u$ and $s_v$ of $u$ and $v$, respectively, we go over all the neighbours of $s_u$ and $s_v$ and determine for each one, according to its degree, whether it belongs to the cluster of $s_u$ and $s_v$, respectively.
We then find the subsets that $u$ and $v$ belong to as defined in Steps~\ref{step:prev} and~\ref{step:prev2} and return YES if and only of $\{u, v\}$ is the edge of minimum rank that connects these subsets.

The above three cases cover all possible scenarios which implies that the time (and probe) complexity of the local implementation of Algorithm~\ref{algGlobal:spanner5} is $\tO(n^{2/3})$ as claimed.
\end{proof}

}

\ifnum\confver=0
\localimp5
\input{figures/fig5rep.tex}
\fi

\section{LCA for constructing \texorpdfstring{$O(k^2)$}{O(k\^2)}-spanners}\label{sec:spanner}

In this section, we present our LCA for constructing $O(k^2)$-spanners. 

\subsection{The algorithm that works under a promise}\label{sec:promise}
We begin by describing a global algorithm for constructing an $O(k^2)$-spanner which works under the following promise on the input graph $G = (V, E)$.
Let $L \eqdef cn^{1/3} \log n$, where $c$ is a constant that will be determined later.
For every $v \in V$, let $i_v \eqdef \min_r \{|\Gamma_r(v)| \geq L\}$. 
We are promised that $\max_{v\in V} \{i_v\} \leq k$. In words, we assume that the $k$-hop neighborhood of every vertex in $G$ contains at least $L$ vertices. 

In addition, we assume without loss of generality that $k = O(\log n)$ as already for $k=\log n$ our construction yields a spanner with $\tO(n)$ edges on expectation.

Our algorithm builds on the partition of $V$ which is described next.

\subsection{The Underlying Partition}\label{sub.par}

\subparagraph*{Centers.}
Pick a set \(S \subset V\) by independently including each vertex \(v\) in $S$ with probability $n^{-1/3}\log n$, so that \(|S| = \Theta(n^{2/3}\log n)\) w.h.p.
We shall refer to the vertices in $S$ as {\em centers}.  
For each vertex $v \in V$, its {\em center}, denoted by $c(v)$, is the center which is closest to $v$ amongst all centers (break ties between centers according to the id of the center).

\subparagraph*{Voronoi cells.}
The {\em Voronoi cell} of a vertex $v$, denoted by $\Vor(v)$, is the set of all vertices $u$ for which $c(u) = c(v)$.
Additionally, we assign to each cell a random rank, so that there is a uniformly random total order on the cells;
note carefully that the rank of a cell thus differs from the rank of its center (which is given by its identifier, which is not assigned randomly).
We remark that we can determine the rank of the cell from the shared randomness and the cell's identifier, for which we simply use the identifier of its center. 
\subsubsection{Clusters}\label{sec:clusters}
The Voronoi cells are partitioned into clusters which are classified into a couple of categories as described next.

\ifnum\confver=0
\input{figures/figheavy.tex}
\fi

\subparagraph*{Singleton Clusters.}
For each Voronoi cell, consider the BFS tree spanning it, which is rooted at the respective center. 
For every $v\in V$, let $p(v)$ denote the {\em parent} of $v$ in this BFS tree. If $v$ is a center then $p(v) = v$.
For every $v\in V\setminus S$, let $T(v)$ denote the subtree of $v$ in the above-mentioned BFS tree when we remove the edge $\{v, p(v)\}$;
for $v\in S$, $T(v)$ is simply the entire tree.
Now consider a Voronoi cell.
If the cell contains at most $L$ vertices, then the {\em cluster} of all the vertices in the Voronoi cell is the cell itself.
Otherwise, there are two cases. 
If $T(v)$ contains more than $L$ vertices, then we say that $v$ is {\em heavy} and define the cluster of $v$ to be the singleton $\{v\}$.
Otherwise, we say that $v$ is {\em light} and its cluster is defined as follows.

\ifnum\confver=0
\input{figures/figsingle.tex}
\fi

\subparagraph*{Non-singleton clusters.}
Observe that if $v$ is light then it has a unique ancestor $u$ (including $v$) such that $u$ is not heavy and $p(u)$ is heavy. 
We define the cluster of $v$ to consist of $T(u)$ and possibly additional subtrees, $T(u')$, where $u'$ is a also a child of $p(u)$ (in $T(p(u)$), as described next (for illustration, see Figure~\ref{fig:light-v}).

We begin with some definitions and notations.
In order to determine the cluster of $u$ (which is also the cluster of $v$) consider transforming the heavy vertex $r = p(u)$ into a binary tree which we call {\em the auxiliary tree of $r$}, $B_r$, as follows.
$B_r$ is rooted at $r$ and has $i$ complete layers where $i$ is such that $2^{i} < \deg(r)$ and $2^{i+1} \geq \deg(r)$. These layers consist of {\em auxiliary vertices}, namely they do not correspond to vertices in $G$. We then add another layer to $B_r$ consisting of the neighbors of $r$, sorted from left to right according to their index in $N(r)$.
Note that except from the root and the vertices at the last layer of $B_r$, all vertices in $B_r$ are auxiliary vertices. This complete the definition of $B_r$ (see Figure~\ref{fig:singleton} for illustration).
For each vertex $x \in B_r$ we define $B_r(x)$ to be the subtree of $B_r$ rooted at $x$. 
We define $S(x) \eqdef B_r(x) \cap N(r)$, namely $S(x)$ is the set of vertices of $N(r)$ which are in the subtree of $B_r$ rooted at $x$. 
The descendants of $x$, denoted by the set $D(x)$, are defined to be the union of the vertices in $T(y)$ for every $y \in S(x)$, namely $D(x) \eqdef \bigcup_{y\in S(x)} T(y)$. 
The weight of $x$ is defined to be the number of vertices in $D(x)$, namely, $w(x) \eqdef |D(x)|$.

We are now ready to define the cluster of $u$.
Let $z(u)$ be the unique ancestor of $u$ in $B_r$ (including $r$), $z$,  for which $w(z) \leq L$ and $w(p(z)) > L$ (where $p(z)$ denotes the parent of $z$ in $B_r$).
The cluster of $u$ (and $v$) is defined to be the set $D(z)$. This completes the description of how the Voronoi cell is partitioned into clusters.

\subparagraph*{Special vertices.}
In order to bound the number of clusters (see Section~\ref{sec:sparsity}) we shall use the following definitions.

\begin{definition}[Special vertex]\label{def:special}
We say that a vertex $u$ is {\em special} if $|T(u)| > L$ and for every child of $u$ in $T(u)$, $t$, it holds that $|T(t)| \leq L$.
\end{definition}

Analogously we define special auxiliary vertex as follows. 
\begin{definition}[Special auxiliary vertex]\label{def:auxspec}
We say that an auxiliary vertex $y$ is a {\em special auxiliary vertex} if either of the following conditions hold: 
\begin{compactenum}
\item\label{auxitem1} $y$ is a parent of a (non auxiliary) vertex $v$ which is heavy. In this case we say that $y$ is a type (a) special vertex.
\item\label{auxitem2} $w(y) > L$ and for every child of $y$, $t$, it holds that $w(t) \leq L$. In this case we say that $y$ is a type (b) special vertex.
\end{compactenum}
\end{definition}
See Figure~\ref{fig:special_cluster} for illustration of constructing cluster from \emph{special auxiliary vertex}.
\newline
\medskip
For a cluster $C$, let $c(C)$ denote the center of the vertices in $C$ (all the vertices in the cluster have the same center).
Let $\Vor(C)$ denote the Voronoi cell of the vertices in $C$.

\medskip\noindent 
This describes a partition of $V$ into Voronoi cells, and a refinement of this partition into clusters.

\ifnum\confver=0
\input{figures/figspecial.tex}
\fi

\subsection{The Edge Set}\label{subsec.edge}
Our spanner, $G' = (V, E')$, initially contains, for each Voronoi cell, $\Vor$, the edges of the BFS tree that spans $\Vor$, i.e., the BFS tree rooted at the center of $\Vor$ spanning the subgraph induced by $\Vor$.
Clearly, the spanner spans the subgraph induced on every Voronoi cell.
Next, we describe which edges we add to $E'$ in order to connect adjacent clusters of different Voronoi cells.
\subsubsection*{Marked Clusters and Clusters-of-Clusters}
Each center is {\em marked} independently with probability $p \eqdef 1/n^{1/3}$. 
If a center is marked, then we say that its Voronoi cell is marked and all the clusters in this cell are marked as well.

\subparagraph*{Cluster-of-clusters.}
For every marked cluster, $C$, define the {\em cluster-of-clusters} of $C$, denoted by $\calC(C)$, to be the set of clusters which consists of $C$ and all the clusters which are adjacent to $C$. 
Let $B$ be a non-marked cluster which is adjacent to at least one marked cluster.
Let $Y$ denote the set of all edges such that one endpoint is in $B$ and the other endpoint belongs to a marked cluster. 
The cluster $B$ is {\em engaged} with the marked cluster $C$ which is adjacent to $B$ and for which the edge of minimum rank in $Y$ has its other endpoint in $C$.

\subsubsection*{The Edges between Clusters}
By saying that we {\em connect} two adjacent subsets of vertices $A$ and $B$, we mean that we add the minimum ranked edge in $E(A, B)$ to $E'$.
For a cluster $A$, define its {\em adjacent centers} $\Cen(\partial A)\eqdef \{c(v)\,|\,u\in A \wedge \{u,v\}\in E \}\setminus \{c(A)\}$, i.e., the set of centers of Voronoi cells that are adjacent to $A$.
This definition explicitly excludes $c(A)$, as there is no need to connect $A$ to its own Voronoi cell.

We next describe how we connect the clusters. The high-level idea is to make sure that for every adjacent clusters $A$ and $B$ we connect $A$ with the cluster engaging $B$ (perhaps not directly) and vise versa.
For clusters which are not adjacent to any marked cluster and hence not engaged with any cluster we make sure to keep them connected to all adjacent Voronoi cells. Formally:

\begin{compactenum}
\item We connect every cluster to every adjacent marked cluster.\label{connect_marked}
\item Each cluster $A$ that is not engaged with any marked cluster (i.e., no cell adjacent to $A$ is marked) we connect to each adjacent cell.
\label{connect_no_adjacent}
\item Suppose cluster $A$ is adjacent to cluster $B$, where $B$ is adjacent to a marked cell. 
Denote by $C$ the (unique) marked cluster that $B$ is engaged with. 
We connect $A$ with $B$ if the following conditions hold:
\begin{compactitem}
\item the minimum ranked edge in $E(A, \Vor(B))$ is also in $E(A, B)$
\item $c(B)$ is amongst the $n^{1/k}\log n$  lowest ranked centers in $\Cen(\partial A) \cap \Cen(\partial C)$
\end{compactitem}

\label{connect_indirect}
\end{compactenum}

\ifnum\confver=0
\subsection{Sparsity}\label{sec:sparsity}
\begin{claim}
The number of clusters, denoted by $s$, is at most $|S| + O(n k\log \Delta) /L)$.
\end{claim}
\begin{proof}
We first observe that, due to the promise on $G$, it follows that for every $v \in V$, the distance between $v$ and  $c(v)$ is at most $k$.
Recall the terminology from Subsection~\ref{sub.par}. 

Consider $v$ which is heavy and therefore its cluster is the singleton $\{v\}$.  
By an inductive argument, it follows that $v$ is an ancestor of a special vertex (see Definition~\ref{def:special}).
Since for every pair of special vertices $u$ and $w$, $T(u)$ and $T(w)$ are vertex disjoint, we obtain that there are at most $n/L$ special vertices. Since for every special vertex, there are at most $k$ ancestors, the total number of heavy vertices is bounded by $nk / L$.

Observe that any cluster either (i) is an entire Voronoi cell
(ii) is a singleton $\{v\}$ where $v$ is heavy (iii) is not a singleton and contains a node $v$ such that $p(v)$ is heavy.
The number of type (i) clusters is bounded by the number of Voronoi cells $|S|$.
We just bounded the number of clusters of type (ii) by $nk /L$. 
Thus it remains to bound the number of type (iii) clusters.

Let $A$ be a type (iii) cluster.
Namely, $A$ is not a singleton and contains a node $v$ such that $p(v)$ is heavy. Let $r = p(v)$. We say that the cluster $A$ is {\em assigned} to $r$. 
Since $r$ is a singleton, some of its children may be singletons and the rest of its children belong to a type (iii) cluster which is assigned to $r$. 

Let $u$ be a child of $r$ which belongs to $A$.
As described in Section~\ref{sec:clusters}, the cluster of $u$ is defined to be $D(z(u))$ where $z(u)$ is an auxiliary vertex and the weight of the parent of $z(u)$ in $B_r$, $p$, is greater than $L$.
By an inductive argument, it follows that $p$ is an ancestor (in $B_r$) of an auxiliary special vertex (see Definition~\ref{def:auxspec}). Since the depth of $B_r$ is bounded by $\log \Delta$ we obtain that the number of
vertices in $B_r$ which are parents (in $B_r$) of vertices, $z$, such that $D(z)$ is a cluster is at most $\log \Delta$ times the number of auxiliary special vertices in $B_r$. 
Since $B_r$ is a binary tree it follows that the number of clusters which are assigned to $r$ is bounded by 
$2\log \Delta$ times the number of auxiliary special vertices in $B_r$. 

Let $x$ and $y$ be auxiliary special vertices of type (b). It follows that $D(x)$ and $D(y)$ are disjoint. Thus the number of auxiliary special vertices of type (b) is bounded by $n/L$. 
Therefore the total number of auxiliary special vertices is bounded by the number of heavy vertices, which is at most $nk/L$,  plus $n/L$.

We conclude that the total number of clusters is bounded by $|S| + O((n k\log \Delta) /L)$, as desired.
\end{proof}

\begin{claim}
The number of edges in $E'$ is $O(n^{1+1/k}\cdot k^2 \log^3n)$ with high probability.
\end{claim}
\begin{proof}
The number of edges we add due to the BFS trees of the Voronoi cells is at most $|V|-1$.

The number of edges which are taken due to Condition~\ref{connect_marked} is at most $s$ times the number of marked clusters, denoted by $m$.
The expectation of $m$ is exactly $s\cdot p$. 
Since $s = O(n^{2/3} k \log n)$ and $p = 1/n^{1/3}$ we obtain that the expected number of edges which are taken due to Condition~\ref{connect_marked} is bounded by $s^2 p = O(n k^2 \log^2 n)$ and w.h.p. is $O(n^{1+1/k}\cdot k^2 \log^3n)$.

Observe that the probability that cluster $A$ is not adjacent to a marked cell is $(1-p)^{|\Cen(\partial A)|}\leq e^{-p|\Cen(\partial A)|}$.
Hence, if $|\Cen(\partial A)|\geq 3p^{-1}\ln n$, $A$ is w.h.p. adjacent to a marked cell. 
Using a union bound over all clusters, it follows that with probability at least $1-1/n^2$ each cluster $A$ without an adjacent marked cell satisfies that $|\Cen(\partial A)|\leq 3p^{-1}\ln n$; therefore, w.h.p. the number of edges which are taken due to Condition~\ref{connect_no_adjacent} is bounded by $(3s\ln n)/p = O(n k \log n)$.

Let $A$ be a cluster. The number of edges which are adjacent to $A$ and are taken due to Condition~\ref{connect_indirect} is bounded by the total number of marked clusters times $n^{1/k}\log n$. 
Thus, the total number of edges which are taken due to Condition~\ref{connect_indirect} is bounded by $s \cdot m \cdot n^{1/k}\log n$.

To conclude, w.h.p. the total number of edges in $E'$ is $O(n^{1+1/k}\cdot k^2 \log^3n)$, as desired.
\end{proof}

\subsection{Connectivity and Stretch} 

\begin{claim}\label{lem:adj}
$G'$ is connected.
\end{claim}
\begin{proof}
Recall that $G'$ contains a spanning tree on every Voronoi cell, hence it suffices to show that we can connect any pair of Voronoi cells by a path between some of their vertices.
Moreover, the facts that $G$ is connected and the Voronoi cells are a partition of $V$ imply that it is sufficient to prove this for any pair of adjacent Voronoi cells.
Accordingly, let $\Vor$ and $\Vor_1$ be two cells such that $E(\Vor, \Vor_1) \neq \emptyset$.

Consider clusters $A\subseteq \Vor$ and $B\subseteq \Vor_1$ such that the edge $e$ of minimum rank in $E(\Vor,\Vor_1)$ is in $E(A,B)$.
If $B$ is not adjacent to a marked cell, then Condition~\ref{connect_no_adjacent} implies that $e$ is selected into $H$.
Thus, we may assume that $B$ is adjacent to a marked cell $\Vor'$ such that there exists a marked cluster $C\subseteq \Vor'$ such that $B$ is engaged with $C$.

If the rank of $\Vor_1$ is minimum in $\Vor(\partial C)\cap \Vor(\partial A)$, then $e$ is selected into $G'$ by Condition~\ref{connect_indirect} and we are done.
Otherwise, observe that $\Vor_1$ is connected to $\Vor'$, as the edge of minimum rank in $E(B,C)$ is selected into $G'$ by Condition~\ref{connect_marked}.
Therefore, it suffices to show that $\Vor$ gets connected to $\Vor'$.
Let $\Vor_2$ be the cell of minimum rank among $\Vor(\partial C)\cap \Vor(\partial A)$.
Let $D\subseteq \Vor_2$ be the cluster satisfying that the edge $e'$ of minimum rank in $E(A,\Vor_2)$ is in $E(A,D)$.
Note that $\Vor_2$ is connected to $\Vor'$ (which we saw to be connected to $\Vor_1$), as there is some cluster $D'\subseteq \Vor(D)$ that is adjacent to $C$ and selects the edge of minimum rank in $E(D',C)$ by Condition~\ref{connect_marked}.

Overall, we see that it is sufficient to show that $\Vor$ gets connected to $\Vor_2$, where $\Vor_2$ has smaller rank than $\Vor_1$.
We now repeat the above reasoning inductively.
In step $i$, we either succeed in establishing connectivity between $\Vor$ and $\Vor_i$, or we determine a cell $\Vor_{i+1}$ that has smaller rank than $\Vor_i$ and is connected to $\Vor_i$.
As any sequence of Voronoi cells of descending ranks must be finite, the induction halts after finitely many steps.
Because the induction invariant is that $\Vor_{i+1}$ is connected to $\Vor_i$, this establishes connectivity between $\Vor$ and $\Vor_1$, completing the proof.
\end{proof} 
See Figure~\ref{fig:connectivity} for illustration of connectivity and stretch.

\ifnum\confver=0
\input{figures/figconnect.tex}
\fi

\begin{claim}\label{lem:chain}
Denote by $G_{\Vor}$ the graph obtained from $G$ by contracting Voronoi cells and by $G'_{\Vor}$ its subgraph obtained when doing the same in $G'$.
If the cells' ranks are uniformly random, w.h.p.\ $G'_{\Vor}$ is a spanner of $G_{\Vor}$ of stretch $O(k)$.
\end{claim}
\begin{proof}
Recall the proof of Lemma~\ref{lem:adj}.
We established connectivity by an inductive argument, where each step increased the number of traversed Voronoi cells by two.
Hence, it suffices to show that the induction halts after $O(k)$ steps w.h.p.

To see this, observe first that $G_{\Vor}$ is independent of the ranks assigned to Voronoi cells and pick any pair of adjacent cells $\Vor$, $\Vor_1$, i.e., neighbors in $G_{\Vor}$.
We perform the induction again, assigning ranks from high to low only as needed in each step, according to the following process.
In each step, we query the rank of some cells, and given an answer of rank $r$, the ranks of all cells of rank at least $r$ are revealed as well.  
In step~$i$, we begin by querying the rank of $\Vor_i$.
Consider the cluster $D_i \subseteq \Vor_i$ adjacent to $A$ satisfying that the edge with minimum rank in $E(\Vor_i, A)$ is also in $E(D_i, A)$.
We can assume without loss of generality that $D_i$ is engaged with a marked cluster $F_i$ (as otherwise $D_i$ connects to $A$ directly and we can terminate the process).
If the rank of anyone of $n^{1/k}\log n$ lowest ranked cells which are adjacent to both $F_i$ and $A$ was already revealed, then the process terminates. 
Otherwise, we query the rank of all the cells which are adjacent to both $F_i$ and $A$ whose rank is still unrevealed.
We set the cell of the queried cluster that has minimum rank to be $\Vor_{i+1}$ and we continue to the next step.

We claim that, in each step $i$, either the process terminates, or the rank of $\Vor_{i+1}$ is at most $1/n^{1/k}$ of the rank of $\Vor_i$ with high probability.
To verify this, observe that in the beginning of step $i$, any cell center whose rank was not revealed so far has rank which is uniformly distributed in $[r_i-1]$, where $r_i$ is the rank of $\Vor_i$.\footnote{In step $1$, we first query $\Vor_1$ and then observe that this statement holds.}
Since there are at least $n^{1/k} \log n$ such cells, we obtain that with high probability the rank of the minimum ranked cell is at most $r_i/n^{1/k}$, as desired.
Hence, with high probability the process terminates after $O(k)$ steps as $r_1$ is bounded by the number of Voronoi cells, which itself is trivially bounded by $n$.
By the union bound over all pairs of cells $\Vor$ and $\Vor_1$, we get the desired guarantee. 
\end{proof}

\begin{claim}\label{cor:chain}
W.h.p., $G'$ is a spanner of $G$ of stretch $O(k^2)$.
\end{claim}
\begin{proof}
Due to the promise on $G$, w.h.p.\ the spanning trees on Voronoi cells have depth $O(k)$.
Hence, the claim holds for any edge within a Voronoi cell.
Moreover, for an edge connecting different Voronoi cells, by Lemma~\ref{lem:chain}, w.h.p.\ there is a path of length $O(k)$ in $G'_{\Vor}$ connecting the respective cells.
Navigating with at most $O(k)$ hops in each traversed cell, we obtain a suitable path of length $O(k^2)$ in $G'$. 
\end{proof}

\subsection{The algorithm for general graphs}\label{sec:general}
We use a combination of the algorithm in Section~\ref{sec:spanner} with the algorithm by Baswana and Sen~\cite{BS07} which has the following guarantees.

\begin{theorem}[\cite{BS07}]\label{thm:BS}
There exists a randomized $k$-round distributed algorithm for computing a $(2k-1)$-spanner $G' = (V,E')$ with $O(kn^{1+1/k})$ edges for an unweighted input graph $G = (V, E)$. More specifically, for every $\{u, v\} \in E'$, at the end of the $k$-round procedure, at least one of the endpoints $u$ or $v$ (but
not necessarily both) has chosen to include $\{u, v\}$ in $E'$. 
\end{theorem}

We call a vertex $v$ {\em remote} if the $k$-hop neighborhood of $v$ contain less than $L$ vertices. 
We denote by $\bar{R}\eqdef V\setminus R$ the set of vertices which are not remote.

\subparagraph*{First Step.} Run the algorithm from Section~\ref{sec:spanner} on the subgraph induced by $\bar{R}$, i.e., $\{u,v\}\in E$ with $u,v\in \bar{R}$ is added to $E'$ if and only if the algorithm outputs the edge.

\subparagraph*{Second Step.} Run the algorithm of Baswana and Sen~\cite{BS07} on the subgraph $H = (V, \{\{u,v\}\in E \,|\, u\in R \text{ or } v \in R \})$, i.e., $\{u,v\}\in E$ with $u\in R$ or $v \in R$ is added to $E'$ if and only if the algorithm outputs the edge.\footnote{The algorithm is described for connected graphs; we simply apply it to each connected component of $H$.}

\subsection{Stretch Factor}
Consider any edge $e=\{u,v\}\in E\setminus E'$ we removed.
If both $u$ and $v$ are in $\bar{R}$, then $e$ was removed by the Algorithm from Section~\ref{sec:spanner}, which was applied to the subgraph induced by $\bar{R}$.
Applying Claim~\ref{lem:chain} to the connected component of $e$, we get that w.h.p.\ there is a path of length $O(k^2)$ from $u$ to $v$ in $G'$.
If $u$ or $v$ are in $R$, by Theorem~\ref{thm:BS} there is a path of length $O(k)$ from $u$ to $v$ in $G'$.
\begin{corollary}
The above algorithm guarantees stretch $O(k^2)$ w.h.p.\ and satisfies that the expected number of edges in $E'$ is $O(n^{1+1/k}\cdot k^2 \log^3n)$
\end{corollary}

\subsection{The local implementation}\label{sec:localC}
In this section we prove the following theorem.

\begin{theorem}\label{thm:k2-main}
There exists an LCA that given access to an $n$-vertex simple undirected graph $G$, with high probability constructs a $O(k^2)$-spanners with $\tO(n^{1+1/k})$ edges in expectation. The probe complexity and time complexity are $O(n^{2/3}\Delta^2)$. Moreover, the algorithm access the graph only by $\allnbr$ queries (and performs $O(n^{2/3}\Delta)$ such queries).
\end{theorem}

\begin{algorithm}[h]
\caption{LCA for constructing $O(k^2)$-spanners}\label{alg:main}
\textbf{Input:} $\{u, v\} \in E$\\
\textbf{Output:} whether $\{u, v\}$ is in $E'$ or not.
\begin{enumerate} 
\item If $u$ or $v$ are in $R$, simulate the algorithm of Baswana and Sen at $u$ and $v$ when running it on the connected component of $u$ and $v$ in the subgraph $H$ (see Section~\ref{sec:general}). 
Return {\bf YES} if either $u$ or $v$ has chosen to include $\{u, v\}$ and {\bf NO} otherwise.\label{s1}
\item \label{s2} Otherwise, $u,v\in \bar{R}$ and we proceed according to Section~\ref{sec:promise}, where all nodes in $R$ are ignored:
\begin{enumerate}
\item If $\Vor(u)  = \Vor(v)$, return {\bf YES} if $\{u, v\}$ is in the BFS tree of $\Vor(u)$ and {\bf NO} otherwise.  

\item Otherwise, let $Q$ and $W$ denote the clusters of $u$ and $v$, respectively.
Return {\bf YES} if at least one of the following conditions hold for $A = Q$ and $B = W$, or symmetrically, for $A = W$ and $B = Q$, and {\bf NO} otherwise. 

\begin{enumerate}
\item \label{con.0} $A$ is a marked cluster and $\{u,v\}$ has minimum rank amongst the edges in $E(A,B)$.

\item \label{con.1} $A$ is not engaged with any marked cluster. Namely, all the clusters which are adjacent to $A$ are not marked. 
In this case, we take $\{u, v\}$ if it has minimum rank amongst the edges in $E(A, \Vor(B))$.

\item \label{con.2} There exists a marked cluster $C$ such that $B$ is engaged with $C$, and the following holds:
\begin{itemize}
\item $\{u,v\}$ has minimum rank amongst the edges in $E(A, \Vor(B))$.  
\item The cell $\Vor(B)$ is amongst the $n^{1/k}\log n$ minimum ranked cells in $\Cen(\partial A) \cap \Cen(\partial C)$
\end{itemize}
\end{enumerate}
\end{enumerate}
\end{enumerate} 
\end{algorithm}

\begin{proof}
The local implementation of the algorithm which is described in the previous section is listed in Algorithm~\ref{alg:main}.
The correctness of the algorithm follows from the previous sections. We shall prove that its complexity is as claimed.

\subparagraph*{The local implementation for remote vertices.}
For Step~\ref{s1}, we need to determine for both $u$ and $v$ if they are remote. 
Recall that a vertex $u$ is remote if its $k$-hop neighborhood contains less than $L$ vertices.
Therefore, we can decide for any vertex $u$ whether it is in $R$ with at most $L$ $\allnbr$ probes. Thus the probe and time complexity is $O(L\Delta)$.
If either $u$ or $v$ are remote then we need to determine for each vertex in their $k$-hop neighborhood whether it is remote or not. 
If either $u$ or $v$ are remote then the $k$-hop neighborhood of each of them contain at most $L \Delta$ vertices. This follows from the fact that the size of the $k$-hop neighborhood of $v$ is at most $\Delta$ factor bigger from the $k$-hop neighborhood of $u$ and vice versa.
Thus, we need to call $\allnbr$ at most $L^2\Delta$ times for this step.
Hence, we obtain that the probe and time complexity of this step is $O(L^2 \Delta^2)$, in total.

\medskip
If $u,v\in \bar{R}$, namely, when both $u$ and $v$ are non-remote, the algorithm proceeds as in Section~\ref{sec:promise}. We next describe the local implementation of the algorithm for this case. 

\subparagraph*{Finding the center and reconstructing the BFS tree.}
We first analyse the probe and time complexity of determining the center of a vertex.
Given a vertex $v$ we perform a BFS from $v$ layer by layer and stop at the first layer in which we find a center or after exploring at least $L$ vertices.
Let $i$ denote the layer in which the execution of the BFS stops. It follows that up to layer $i-1$ we explored strictly less than $L$ vertices. 
Thus this step can be implemented by $O(L)$ calls to $\allnbr$.
In particular, the probe and time complexity of finding the center is $O(L\Delta)$.

We observe that at the same cost we also determine the path from $c(v)$ to $v$ in the BFS tree rooted at $c(v)$ as follows. The parent of $v$ in the tree is the neighbour of $v$ that has minimum id amongst all neighbour of $v$ that are closer than $v$ to $c(v)$. Similarly, we can determine the parent of the parent of $v$ and so on until we reach $c(v)$.

\subparagraph*{Determining if a vertex is heavy.}
In order to reconstruct the clusters we need to be able to determine if a vertex is heavy or not. Recall that a vertex $v$ is heavy if $|T(v)| > L$.
We explore $T(v)$ by performing a find center procedure on all the neighbours of $v$ and then continuing recursively on all the neighbours of $v$ that belong to $\Vor(v)$.
Since finding the center takes $O(L)$ calls to $\allnbr$, we conclude that we can determine whether $v$ is heavy or light by using $O(L^2\Delta)$ calls to $\allnbr$. This follows from the fact that when we partially or completely reveal $T(v)$, we need to find the center of at most $L\Delta$ vertices. Thus, the overall probe and time complexity for this step is $O(L^2\Delta^2)$.

\subparagraph*{Reconstructing the clusters.}
Given a vertex $v$ we reconstruct its cluster as follows. First, perform a find-center operation on $v$ and let $v:=u_0,u_1,\ldots,u_d$ be the path to the center. We then determine if $v$ is heavy using the prior procedure (and if so we are done). Otherwise, iteratively find $T(u_i)$ for $i\in [d]$ (where we do not search down the path that we have already explored), terminating the search when $T(u_i)>L$. In this case, construct the tree of special vertices below $u_i$ and again find the first ancestor of $u_{i-1}$ in this tree that is heavy, and let the cluster be the children of the predecessor special vertex. As we ultimately explore only $O(L\Delta)$ vertices, this results in  $O(L\Delta$) calls to find-center, which results in at most $O(L^2\Delta)$ calls to $\allnbr$.

\subparagraph*{Determining the cells adjacent to clusters.}
For Step~\ref{s2} we need to reconstruct the cluster of $u$, the cluster of $v$, and the clusters that $u$ and $v$ are engaged with; this takes $O(L^2\Delta)$ calls to $\allnbr$. In addition, for each of these clusters $C$, we need to determine the center of each vertex adjacent to $C$. Since the size of the clusters is bounded by $L$, the number of vertices adjacent to $C$ is at most $L \Delta$.
Therefore the number of calls to find-center is at most $L \Delta$.
This likewise requires $O(L^2\Delta)$ calls to $\allnbr$ and overall $O(L^2\Delta^2)$ probes and time.   

\medskip
We conclude that we can perform all necessary checks to decide whether $\{u,v\}\in E'$ or not using $O(L^2\Delta)$ calls to $\allnbr$ which invokes $O(L^2\Delta^2)$ neighbour probes. By the analysis above, the time complexity is $O(L^2\Delta^2)$ as well.
\end{proof}
\fi

\section{Graph Decomposition Via Ranking}\label{sec:decomp}
We give the formal statement of the LCA of Item~\ref{item4}. We note that our decomposition gives a stronger promise than the maximum degree of each subgraph being bounded, in that we actually bound the degree vertex-wise. In particular, up to poly-logarithmic factors, the average degree of every subgraph is equal to the overall average degree divided by the number of subgraphs with high probability.
\begin{theorem}\label{thm:decomp}
    There exists an LCA that, given a parameter $R \leq \sqrt{\Delta}$ and access to an $n$-vertex simple undirected graph $G=(V,E)$ with maximum degree $\Delta$,
    decomposes $G$ into edge-disjoint subgraphs $G_1,\ldots,G_{R^2}$ such that:
    \begin{enumerate}
        \item Given $(u,v) \in E$, the $i\in [R^2]$ such that $(u,v)\in G_i$ can be computed in time and space $O(1)$.
        \item Given $v\in V$ and $i\in [R^2]$, $\allnbr_i(v)$ can be computed in time and space $O(d_G(v)/R)$.
        \item With high probability, the degree of $v$ in $G_i$ is $O(\max\{\log(n),d_G(v)/R^2\})$ for every $v\in V$ and $i\in [R^2]$. In particular, for every $i$ the maximum degree of $G_i$ is $O(\max\{\log(n),\Delta/R^2\})$ with high probability.
    \end{enumerate}
\end{theorem}
We first describe the decomposition in a global manner. We refer to each subgraph as a color. We assign each edge in $G$ one of $R^2$ colors, such that the degree and $\allnbr$ query times are as claimed. We will identify the set of colors with $[R] \times [R]$, and assume $R^2 \le \Delta$ since otherwise the statement is trivial. For convenience, let $d(v):=d_G(v)$ and $d_i(v):=d_{H_i}(v)$.

Each vertex $v$ draws a random value $r_v \sim [R]$. Furthermore, for every vertex $v$, let the first $\lceil d(v)/R\rceil$ neighbors of $v$ be $B_1(v)$, the second be $B_2(v)$, etc. Note that this divides the out-edges into $R$ blocks. 
For an edge $(u,v)$ with $u<v$ in blocks $B_i(u),B_j(v)$ respectively, the color of the edge is the pair $(i+r_u \mod R, \;\; j+r_v \mod R).$
Let $G_{a,b}$ for $a,b\in [R]$ be the subgraph consisting of all edges with color $(a,b)$. 

We can then combine Theorem~\ref{thm:decomp} with Theorem~\ref{thm:k2-main} to give the final result.
\begin{corollary}\label{cor:k2-main}
    There exists an LCA that given access to an $n$-vertex simple undirected graph $G$ with maximum degree $\Delta$, constructs an $O(k^2)$-spanner with $\tO(n^{1+1/k})$ edges whose probe complexity and time complexity are $O(n^{2/3-(1.5-\alpha)/k}\Delta^2)$, for any constant $\alpha >0$.
\end{corollary}

\ifnum\confver=0
\subsection{Decomposition Implementation}
 \begin{claim}
     Given $(u,v)\in G$, we can determine the color $(a,b)$ of the edge in time and space $O(1)$.
 \end{claim}
 \begin{proof}
    By making two adjacency probes, we can determine the indices of edge $(u,v)$ in $u$ and $v$. Then we can compute which blocks contain this edge using two degree queries and a constant number of arithmetic operations, and then compute the final color by looking up the random shifts of $u$ and $v$.
 \end{proof}
	
\begin{claim}\label{clm:degreebound}
    With high probability, for every $v\in V$ and $(a,b)\in [R]\times[R]$ we have $d_{a,b}(v)=O(\max\{\log(n),d(v)/R^2\})$.
\end{claim}	
\begin{proof}=
    Fix an arbitrary vertex $v\in V$ and color $(a,b)\in [R]\times[R]$. Fix its shift of $r_v$ of $v$ arbitrarily. Let $S = B_{a - r_v}(v) \cup B_{b - r_v}(v)$ be the set of all edges incident to $v$ (in $G$) in blocks $a-r_v$ and $b-r_v$. Then $S$ is a superset of the set of edges incident to $v$ with color $(a,b)$, and $|S| = 2d(v)/R$.

    For an arbitrary edge $e = (v,u)$ in $S$, let $X_e$ be the indicator random variable which is 1 exactly when the color of $e$ is $(a, b)$.
    Let $k$ be the block index of $e$ in $u$ so that $e \in B_k(u)$.  There are four cases to consider regarding $e$:
    \begin{itemize}
    \item Case 1: $v<u$ and $e \in B_{a-r_v}(v)$.  Then $e$ has color $(a, b)$ if and only if $r_u$ is equal to $b - k$, which occurs with probability exactly $1/R$.
    \item Case 2: $v < u$ and $e \notin B_{a-r_v}(v)$.  In this case $e$ never has color $(a, b)$.
    \item Case 3: $v>u$ and $e \in B_{b-r_v}(v)$.  Similarly to case 1, $e$ has color $(a, b)$ if and only if $r_u$ is equal to $a - k$, which occurs with probability $1/R$.
    \item Case 4: $v>u$ and $e \notin B_{b-r_v(v)}$.  Similarly to case 2, $e$ never has color $(a, b)$.
    \end{itemize}
    In any case, we have $P[X_e = 1] \le 1/R$ for all $e \in S$.
    Furthermore, the variables in $\{X_e\}_{e \in S}$ are independent random variables: for distinct edges $e, e' \in S$ where $e = (u, v)$ and $e=(u', v)$, $X_e$ and $X_{e'}$ are independent since the random variables $r_u, r_{u'}$ are independent.
 
    Letting $X = \sum_{e \in S} X_e$ and picking an arbitrary constant $c \ge 2$, we find $$E[X] \le 2d(v)/R^2 \le c\max\{\log n,d(v)/R^2\} =: \mu.$$
    By the multiplicative Chernoff's bound we have
    $$\Pr[X > 3\mu] \le \exp(-\mu) \le \exp(-c\log n) = n^{-c},$$
    and so the total number of neighbors of $v$ with color $(a,b)$ is $O(\max\{\log n, d(v)/R^2\})$ with high probability.  Finally, a union bound over all $n$ vertices and $R^2 \le \Delta$ colors completes the proof.
\end{proof}

\begin{claim}\label{clm:allnbrtime}
    $\allnbr_i(v)$ can be computed in time and space $O(d(v)/R)$.
\end{claim}
\begin{proof}
    Given $v\in V$ and a color $i=(a,b)$, let $S = B_{a-r_v}(v) \cup B_{b-r_v}(v)$ as before. Note that these correspond to the blocks which have received labels $a$ and $b$ respectively, given the random shift of vertex $v$.  We make $2d(v)/R$ neighbor probes to determine all elements of \(S\), then $2d(v)/R$ adjacency probes to determine the indices of every edge in the other endpoint. Then for each edge, we can check in time $O(1)$ (by examining the random shift of the other endpoint) if the label is $(a,b)$.
\end{proof}

\begin{proof}[Proof of Corollary~\ref{cor:k2-main}]
    Let $\beta \in (0,1]$ be such that $3/(2+\beta)=1.5-\alpha$.
    Given a query if edge $(u,v)\in G$ is in the spanner, we apply the LCA of Theorem~\ref{thm:decomp} with parameter\footnote{We have $R^2 \in O(n^{1/k})$; in order to apply Theorem~\ref{thm:decomp} this should be at most $\Delta$.  We can assume this since if $\Delta$ is $O(n^{1/k})$ then the graph is already sparse to begin with.} \(R = \left\lceil n^{1/(k(2 + \beta))} \right\rceil\) to determine the $i$ such that $(u,v)\in G_{i}$. We then apply Theorem~\ref{thm:k2-main} with parameter \(k' = \left\lceil k(1 + \frac{2}{\alpha})\right\rceil\) to the graph $G_{i}$ and query if $(u,v)$ is contained in the spanner, and return the answer. Note that we ultimately obtain a \(O(k'^2) = O(k^2)\)-spanner for every subgraph (and thus for the overall graph), and the number of edges is bounded as
    \[\tO(R^2n^{1 + 1/k'}) \le \tO\left(n^{1+\frac{2}{k(2 + \alpha)} + \frac{1}{k(1 + 2/\alpha)}}\right) = \tO(n^{1 + 1/k}).\] 
    Furthermore, the time complexity  is
    \[O(n^{2/3}\Delta^2R^{-3}) \le O\left(n^{\frac23 - \frac{1.5-\alpha}{k}}\Delta^2\right). \qedhere\]
\end{proof}
\fi

\bibliographystyle{plain}
\bibliography{refs}

\end{document}